\documentclass{article}[11pt]
\def\draft{1}
\def\doubleblind{0}
\usepackage[margin=1in]{geometry}
\usepackage[utf8]{inputenc}
\usepackage{amsthm}
\usepackage{amsthm, amsmath, amssymb, mathtools}
\usepackage{bm,bbm}
\usepackage{graphicx}
\usepackage{color,xcolor}
\usepackage{paralist,enumitem}
\usepackage{mathrsfs}
\usepackage[normalem]{ulem}
\usepackage{tabularx}
\usepackage{tikz}
\usepackage{caption,comment}
\usepackage{algorithmicx}
\usepackage{algorithm}
\usepackage{algpseudocode}
\newcounter{algsubstate}
\renewcommand{\thealgsubstate}{\alph{algsubstate}}

\algnewcommand\algorithmicinput{\textbf{Input:}}
\algnewcommand\Input{\item[\algorithmicinput]}

\algnewcommand\algorithmicoutput{\textbf{Output:}}
\algnewcommand\Output{\item[\algorithmicoutput]}

\algnewcommand\algorithmicgoal{\textbf{Goal:}}
\algnewcommand\Goal{\item[\algorithmicgoal]}

\usepackage[colorlinks=true, allcolors=blue]{hyperref}
\usepackage[capitalize,noabbrev,nameinlink]{cleveref}

\newcommand{\vnote}[1]{\ifnum\draft=1\textcolor{orange}{[\textbf{Santhoshini:} #1]}\fi}
\newcommand{\mnote}[1]{\ifnum\draft=1\textcolor{red}{[\textbf{Huacheng:} #1]}\fi}

\newcommand{\Exp}{\mathbb{E}}

\newcommand{\poly}{\mathrm{poly}}
\newcommand{\E}{\mathbb{E}}
\newcommand{\s}{\mathcal{S}}
\newcommand{\cH}{\mathbf{H}}
\newcommand{\cM}{\mathbf{M}}

\newcommand{\checkfn}{\mathsf{check}}
\newcommand{\maxl}{\mathsf{maxL}}
\usepackage{amsthm}
\usepackage{thmtools,thm-restate}
\numberwithin{equation}{section}

\usepackage[colorlinks=true, allcolors=blue]{hyperref}
\usepackage[capitalize,noabbrev,nameinlink]{cleveref}

\declaretheoremstyle[bodyfont=\it,qed=\qedsymbol]{noproofstyle}

\declaretheorem[name=Observation,numbered=no]{observation*}

\declaretheorem[numberlike=equation]{theorem}

\declaretheorem[name=Theorem,numbered=no]{theorem*}

\declaretheorem[numberlike=equation]{lemma}
\declaretheorem[name=Lemma,numbered=no]{lemma*}

\declaretheorem[name=Corollary,numbered=no]{corollary*}

\declaretheorem[numberlike=equation]{proposition}
\declaretheorem[name=Proposition,numbered=no]{proposition*}

\declaretheorem[numberlike=equation]{claim}
\declaretheorem[name=Claim,numbered=no]{claim*}

\declaretheorem[name=Conjecture,numbered=no]{conjecture*}

\declaretheorem[name=Question,numbered=no]{question*}

\declaretheoremstyle[bodyfont=\it]{defstyle} 

\declaretheorem[numberlike=equation,style=defstyle]{definition}
\declaretheorem[unnumbered,name=Definition,style=defstyle]{definition*}

\declaretheorem[unnumbered,name=Notation=defstyle]{notation*}

\declaretheorem[unnumbered,name=Construction,style=defstyle]{construction*}

\declaretheoremstyle[]{rmkstyle}

\declaretheorem[unnumbered,name=Example,style=rmkstyle]{example*}

\usepackage[style=numeric-comp, backend=bibtex, minalphanames=3, maxalphanames=4, maxbibnames=99, sorting=none]{biblatex}

\def\doubleblind{0}
\def\draft{0}

\allowdisplaybreaks
\title{Optimally detecting uniformly-distributed $\ell_2$ heavy hitters in data streams}

\ifnum\doubleblind=0
\author{Santhoshini Velusamy\thanks{University of Waterloo, Waterloo, Ontario, Canada. Email: \texttt{santhoshini.velusamy@uwaterloo.ca}.}
\and Huacheng Yu\thanks{Princeton University, Princeton, New Jersey, USA. Email: \texttt{yuhch123@gmail.com}.}
}
\fi
\ifnum\doubleblind=1
\author{Anonymous authors}
\fi

\date{}
\begin{document}
\maketitle

\begin{abstract}
Given a stream $x_1,x_2,\dots,x_n$ of items from a Universe $U$ of size $\poly(n)$, and a parameter $\epsilon>0$, an item $i\in U$ is said to be an $\ell_2$ heavy hitter if its frequency $f_i$ in the stream is at least $\sqrt{\epsilon F_2}$, where $F_2={\sum_{i\in U} f_i^2}$. Efficiently detecting such heavy hitters is a fundamental problem in data streams and has several applications in both theory and in practice. The classical \textsf{CountSketch} algorithm due to Charikar, Chen, and Farach-Colton [2004], was the first algorithm to detect $\ell_2$ heavy hitters using $O\left(\frac{\log^2 n}{\epsilon}\right)$ bits of space, and their algorithm is optimal for streams with deletions. A follow-up paper due to Braverman, Chestnut, Ivkin, and Woodruff [2016], gave the \textsf{CountSieve} algorithm for insertion-only streams, with an improved space bound. Their algorithm requires only $O\left(\frac{\log(1/\epsilon)}{\epsilon}\log n  \log\log n \right)$ space. Note that any algorithm requires at least $\Omega\left(\frac{1}{\epsilon} \log n\right)$ space to output $O(1/\epsilon)$ heavy hitters in the worst case. So for constant $\epsilon$, the space usage of the \textsf{CountSieve} algorithm is asymptotically only a $\log\log n$ factor worse than the optimum bound. A later work due to Braverman, Chestnut, Ivkin, Nelson, Wang, and Woodruff [2017] gave the \textsf{BPTree} algorithm which gets rid of this additional $\log\log n$ factor and detects $\ell_2$ heavy hitters in insertion-only streams using only $O\left(\frac{\log(1/\epsilon)}{\epsilon}\log n \right)$ space. While their algorithm achieves optimal space bound for constant $\epsilon$, their bound could be sub-optimal for $\epsilon=o(1)$. For \emph{random order} streams, where the stream elements can be adversarial but their order of arrival is uniformly random, Braverman, Garg, and Woodruff [2020] showed that it is possible to achieve the optimal space bound of $O\left(\frac{1}{\epsilon} \log n\right)$ for every $\epsilon = \Omega\left(\frac{1}{2^{\sqrt{\log n}}}\right)$. 

In this work, we generalize their result to \emph{partially random order} streams where only the heavy hitters are required to be uniformly distributed in the stream. We show that it is possible to achieve the same space bound, but with an additional assumption that the algorithm is given a constant approximation to $F_2$ in advance.
We mainly exploit the fact that if the stream is divided into windows of length roughly $\frac{n}{\sqrt{\epsilon F_2}}$, with high probability, for every heavy hitter, every $\log n$ consecutive windows has an occurrence of that heavy hitter, and develop a hierarchical \emph{sample-and-check} framework that efficiently detects them. The window size is crucial to the algorithm. 


\end{abstract}
\newpage
\section{Introduction}

``Heavy hitters'' is a term often used to describe elements that occur most commonly in a data stream. The problem of detecting heavy hitters has several applications, both in theory and practice. For instance, heavy hitters algorithms have been used as subroutines to solve many important streaming problems including norm estimation \cite{Indyk-Woodruff05}, entropy estimation \cite{ChakrabartiCM10,HarveyNO08}, and weighted sampling \cite{MonemizadehW10}. In addition, they have also influenced recent efforts to speed up LLM algorithms. In particular, by detecting and removing
heavy hitters, a recent paper obtained a significant speed up in the matrix multiplication component
of the ``attention approximation problem'' in LLMs \cite{Hyperattention}. 

The earliest known work on heavy hitters began in the 1980's when Boyer and Moore discovered an algorithm to output the majority element in the stream using just two machine words \cite{Boyer1991}.
Since then, vast generalizations of this problem have been studied in the literature (see, for example, \cite{HHref1,HHref2,HHref3,HHref4,HHref5,HHref6,HHref7,HHref8} and the references therein). Formally, given a stream $x_1,\dots,x_n$ of items from a universe $U$ of size $\poly\, n$, the heavy hitters problem asks to recover, with high probability, every element $i\in U$ such that $f_i \ge \gamma$, where $f_i$ denotes the frequency with which the element occurs in the stream. The algorithm is also often required to not output any element whose frequency is less than $\tau$, where the thresholds $\gamma$ and $\tau$ are usually defined in terms of different norms of the stream. 

One popular notion is the $\ell_1$  guarantee where given a parameter $\epsilon\in (0,1]$, the algorithm has to recover every element of frequency at least $\epsilon n$ and not output any element of frequency less than $\delta n$ where $\delta$ is typically assumed to be $\epsilon/c$, for some constant $c$. The first algorithm for this problem was given by Misra and Gries in 1982 \cite{MisraG82}. Their algorithm is deterministic and uses only $O(\log n/\epsilon)$ bits of space, which is optimal since in the worst case, the algorithm has to output $\Theta(1/\epsilon)$ elements and it requires $\Theta(\log n)$ space to represent each of them. Their algorithm also works only in the insertion-only model, also known as the \emph{cash-register} model, where elements can only be inserted into the stream. For general \emph{turnstile} streams where elements can be added to or deleted from the stream, Cormode and Muthukrishnan gave the $\mathsf{CountMin Sketch}$ algorithm which is a randomized \emph{sketching} algorithm\footnote{Sketching algorithms are a special class of streaming algorithms where the algorithm's output is determined by a small sketch it produces of the input stream, and the sketch itself has the property that the sketch of the concatenation of two streams can be computed from the sketches of the two component streams.} that uses $O(\log^2 n/\epsilon)$ bits of space to achieve a constant success probability \cite{CountMinSketch}. Notice the $\log n$ gap in the space required by the Misra-Gries algorithm and the $\mathsf{CountMinSketch}$ algorithm. This gap is inherent as the latter algorithm is known to be optimal for turnstile streams \cite{JowhariST11}.

Another well-studied guarantee is the $\ell_2$ guarantee, where given a parameter $\epsilon\in (0,1]$, the algorithm has to recover every element $i\in U$ such that $f_i^2 \ge \epsilon \Vert f\Vert_2^2$ and not output any element $i\in U$ such that $f_i^2 < \delta \Vert f\Vert_2^2$, where $ \Vert f\Vert_2$ denotes the norm of the frequency vector of the elements in the stream. Note that this is a stronger guarantee compared to the $\ell_1$ guarantee. For example, consider a stream where $\Vert f\Vert_2^2 = \Theta(n)$. For a constant $\epsilon$, the $\ell_1$ guarantee asks only to output elements that occupy a constant fraction of positions in the stream, while the $\ell_2$ guarantee asks for any element that has frequency $\Omega(\sqrt{n})$, which occupies only a negligible portion of the stream. The $\mathsf{CountSketch}$ algorithm due to Charikar, Chen, and Farach-Colton can recover heavy hitters under the $\l_2$ guarantee, with at least a constant probability, using $O(\log^2 n/\epsilon)$ bits of space \cite{CountSketch}. A follow-up work due to Jowhari, Sağlam, and Tardos shows that $\mathsf{CountSketch}$ can be extended to output heavy hitters under the $\ell_p$ guarantee for every $p\le 2$ using space at most $O(\log^2 n/\epsilon)$ and that the algorithm is optimal for every $p$ in this range, in the turnstile setting \cite{JowhariST11}. An earlier work due to Do ba, Indyk, Price, and Woodruff \cite{DBLP:conf/soda/BaIPW10} proved this result for $p=2$.
While $\mathsf{CountSketch}$ is optimal for turnstile streams, it is natural to ask if there is a more space-efficient algorithm in the insertion-only model, similar to the $\ell_1$ case. This question was first answered positively by Braverman, Chestnut, Ivkin, and Woodruff who gave the $\mathsf{CountSieve}$ algorithm which requires only $O\left(\frac{1}{\epsilon}\log\left(\frac{1}{\epsilon}\right)\log n \log\log n\right)$ space to output heavy hitters under the $\ell_2$ guarantee, with at least a constant success probability \cite{CountSieve}. The same authors, along with Nelson and Wang, later gave an improved algorithm called $\mathsf{BPTree}$ which uses only $O\left(\frac{1}{\epsilon}\log\left(\frac{1}{\epsilon}\right)\log n\right)$ space to achieve a constant success probability \cite{BPTree}. 
For \emph{random order} streams where the order of arrival of elements is uniformly random, Braverman, Garg, and Woodruff \cite{BGW20} gave an optimal algorithm that uses only $O(\log n/\epsilon)$ space and achieves constant success probability, for every $\epsilon=\Omega(1/2^{\sqrt{\log n}})$.
The primary goal of this paper is to extend their result to more general settings. In particular, we show that for \emph{partially} random order streams when the heavy hitters are uniformly distributed across the stream (while the non-heavy elements could be adversarially positioned), it is still possible to recover the $\epsilon$-heavy $l_2$ heavy hitters,  using only $O(\log n/\epsilon)$ space with constant success probability, for every $\epsilon=\Omega(1/2^{\sqrt{\log n}})$, assuming have an estimate of the $l_2$ norm of the frequency vector and access to a random oracle.




\subsection{Our result}
We now formally state our main theorem. 



\begin{theorem}[Partially random order streams]\label{thm:main result 2}
 $\exists n_0,c_0\in \mathbb{N}$ such that for every $n\ge n_0$ and $\epsilon\ge c_0\left(\frac{1}{2^{\sqrt{\log n}}}\right)$, there is a single-pass streaming algorithm that with probability at least $9/10$, finds every $\ell_2$ $\epsilon$-heavy hitter and reports no element that is not an  $\ell_2$ $\epsilon/256$-heavy hitter, in partially random order streams of length $n$, using $O(\log n/\epsilon)$ bits of space, assuming that the $\ell_2$ norm of the stream is known in advance to the algorithm and access to a random oracle.
\end{theorem}

The above result should be viewed as a first step towards generalizing \cite{BGW20} beyond random order streams. Assuming that the heavy hitters are uniformly distributed, while the non-heavy elements are placed by an adversary, is a natural starting point. For example, the heavy hitters maybe generated by a stochastic process, while the non-heavy elements could be adversarial noise. Such partially random ordered streams have been explored in the literature (see, for example, \cite{ChiplunkarKK022}) and serve as an important bridge between the random order and the adversarial streaming model. Our other assumption that the streaming algorithm needs to know the $\ell_2$ norm of the stream in advance is not as natural as our first assumption, and we state it this way purely for technical reasons. While one could imagine using an $\ell_2$ norm tracker algorithm like in the $\mathsf{BPTree}$ algorithm to get rid of this assumption, unfortunately if the stream order is not fully random, this could incur an additional factor of $\log(1/\epsilon)$ in space, and hence we would not gain any advantage in space compared to the $\mathsf{BPTree}$ algorithm.\footnote{Such an algorithm that geometrically guesses $\ell_2$ must start the subroutine with the correct guess at the latest, say when there is still half of the stream left. Otherwise, there are not enough occurrences of the heavy hitters to successfully detect them. However, imagine a case where the true $\ell_2$-norm is $F=n^{1.5}$, and there is only one heavy hitter in the stream, which has squared frequency $\epsilon F$. All other (light) elements occur once in the first half. Then we can only rule out the $\ell_2$ smaller than $O(\epsilon F)$ only when we have processed the first half of the stream. Hence, we must maintain $O(\log (1/\epsilon))$ consecutive geometric guesses.}
Therefore, we leave it as an interesting open problem whether one could prove \cref{thm:main result 2} without the assumption about the $\ell_2$ norm of the stream, or if there is a lower bound in this case. We prove \cref{thm:main result 2} at the end of \cref{sec:epsilon-heavy}.







In the following section, we present an overview of our techniques.
\section{Technical overview}
The previous $\ell_2$ heavy hitter algorithms for insertion-only streams (\textsf{CountSieve}, \textsf{BPTree}) first randomly hash the elements into $O(1/\epsilon)$ buckets.
For each $\epsilon$-heavy hitter $x^*$, there is a constant probability that all other elements hashed to the same bucket as $x^*$ have second moment at least a constant smaller than $f_{x^*}^2$.
Such an element $x^*$ is referred to as a \emph{super heavy hitter} in the bucket. 
The algorithms then try to recover a super heavy hitter from each bucket using small space, which will succeed with constant probability.

Since there can be as many as $O(1/\epsilon)$ heavy hitters, in order to recover all of them, the algorithms have to repeat the process independently for $O(\log(1/\epsilon))$ times.
More specifically, this is due to the following two technical reasons:
\begin{itemize}
    \item we must ensure that every $\epsilon$-heavy hitter is super heavy in its bucket (at least once);
    \item then we must ensure that it is recovered (at least once) with probability at least $1-O(\epsilon)$,
\end{itemize}
so that we can union bound over all $O(1/\epsilon)$ heavy hitters.

\textsf{BPTree} uses a subroutine with $O(\log n)$ bits of space to recover a super heavy hitter (with constant probability).
The subroutine is applied on all $O(1/\epsilon)$ buckets each time, and is repeated for $O(\log (1/\epsilon))$ times.
Hence, it has space of $O(\epsilon^{-1}\log (1/\epsilon)\log n)$ bits.
Thus, in order to remove the extra $O(\log(1/\epsilon))$ factor, we must resolve the above two technical issues that forced us to repeat the process. 

In this paper, we first focus on the second issue.
That is, we first assume that an $\epsilon$-heavy hitter $x^*$ is already super heavy ``in its bucket'', and we aim to recover it with high probability, \emph{under the assumption} that $x^*$ is roughly evenly distributed and \emph{with free randomness}.
It turns out that as long as $\epsilon>2^{-\Omega(\sqrt{\log n})}$, our solution for high-probability super heavy hitter generalizes to $\epsilon$-heavy hitter in a natural but different way that does not need to do hashing on top of it.
Hence, for our algorithm, the first issue no longer exists.
We note that there may be simpler solutions for random-order streams, the emphasis of this work is to \emph{design the more general algorithm} that works even only assuming the occurrences of heavy hitters are random.

\subsection{Warmup and summary of~\cite{BGW20}}
To motivate our new algorithm for super heavy hitter, let us first consider a special case of the stream as a warmup.
Consider a stream of length $n$ consisting of a super heavy hitter $x^*$ with frequency $C\sqrt{n}$ for some large constant $C$, and other elements with frequency $1$.
Dividing the stream into \emph{windows} of length $W:=\sqrt{n}/C$, let us assume $x^*$ has one occurrence in every window.
$x^*$ is super heavy in the stream, and we want to find $x^*$ \emph{with high probability} using $O(\log n)$ bits of space.
Note that such partially random order streams actually have random order.
In fact, the algorithm presented below is a (simplification) of the prior work~\cite{BGW20} for random-order streams.\footnote{The prior work~\cite{BGW20} uses a different setting of parameters, and their algorithm for this case succeeds with constant probability instead of with \emph{high} probability. On the other hand, their algorithm does not assume free randomness.}


The main framework to find $x^*$ is \emph{sample-and-check}.
Since $x^*$ is the only element that appears at least twice, as long as we find any duplicated element, the algorithm can declare it to be the heavy hitter.
The algorithm processes the stream in windows.
For each window $i$ (of size $W=\sqrt{n}/C$), it samples a set $S_i\subseteq [U]$ such that each $x\in[U]$ belongs to $S$ independently with probability $q:=1/W=C/\sqrt{n}$.
Hence, by linearity of expectation, window $i$ has expected one total occurrence of elements in $S_i$.
Also, with \emph{constant probability}, there exists a window $i$ such that $x^*\in S_i$ and it is the only occurrence of $S_i$ in window $i$.
A natural idea is to check if a sampled element in window $i$ also appears in window $i+1$, as $x^*$ should.
However, this would require $O(\log n)$ bits of space to remember the identity of the element, and would still only succeed with constant probability.

Our (simplified) algorithm (for this case) uses only $O(\log\log n)$ bits \emph{most of the time} to check it.
We first sample a random hash function $f$ that maps elements in $[U]$ to a range $[K]$ for some $K=\poly\log n$.
As it processes the windows, if for window $i$, there is only one occurrence of elements in $S_i$ (say the element is $x_i$), the algorithm remembers the $\log K$-bit hash value $v=f(x_i)$, rather than $x_i$ itself.
Next, it proceeds to the next window $i+1$, and checks if it also has an element with the same hash value $v$ and in $S_i$ (the sample of window $i$).
If that happens, only then we remember the identity of this element (via the occurrence in window $i+1$), and check if it also appears in window $i+2$.
When $x_i$ is not $x^*$, the expected number of occurrences of elements in $S_i$ in window $i+1$ with the same hash value $v$ is only $1/K$.
Hence, with only $1/K$ probability, do we need to remember the identity of some element using $O(\log n)$ bits, while for the rest of the time, the algorithm only remembers a hash value of $O(\log K)=O(\log\log n)$ bits.

Our next idea is then to \emph{time-share} the $O(\log n)$-space budget across multiple instances of the algorithm, to boost the success probability while restricting the total space usage.
We run $O(\log n/\log K)=O(\log n/\log\log n)$ instances of the above algorithm in parallel.
However, we only allow one instance to use $O(\log n)$ bits at the same time.
That is, if one instance already starts to use $O(\log n)$ bits for checking if the exact element appears in two consecutive windows, then we will ``pause'' all other instances, and only resume when the space is freed. 
That is, all other instances will simply do nothing in the meanwhile, and restart at a later point in the stream.
Each such instance may ``block'' all other instances at most $O(1/K)$-fraction of the time.
Hence, only a small fraction of the time, do we have any instance blocking others, the analysis still goes through: Each instance of the algorithm still finds $x^*$ with constant probability.
Thus, we are able to increase the overall success probability to $1-\exp(-\Omega(\log n/\log \log n))$.

\subsection{General stream with evenly distributed super heavy hitter}
Next, let us consider a general stream of length $n$ and second moment $O(n)$.\footnote{This assumes that we know both the length of the stream and the second moment of the frequency vector. It turns out that if the second moment is $\Theta(nt)$, then sampling each element in the stream with probability $1/t$ reduces to the case with second moment linear in the stream length, and the super heavy hitter remains super heavy.}
Again let us assume that there is a super heavy hitter $x^*$ with frequency $C\sqrt{n}$, and it appears once in every window of length $W=\sqrt{n}/C$ as before.
But, we do not make any assumption on the rest of the stream beyond its length and second moment.

The main framework is the same as the warmup -- sample-and-check.
However, checking whether a sampled element $x$ is super heavy can be more challenging.
For example, some element $x$ with a lower frequency of $C\sqrt{n}/t$ may locally appear like the heavy hitter $x^*$, i.e., it appears like $x^*$ for $1/t$-fraction of the stream.
Thus, the algorithm must spend at least this much time to distinguish such an element from the true heavy hitter, if sampled.
Moreover, in principle, there can be as many as $O(t^2)$ such elements (without breaking the second moment bound).
That is, it is even much more likely to sample such a ``locally heavy'' element than $x^*$ itself.

Our algorithm samples a set $S_i\subseteq [U]$ for each window $i$ and a hash function $f: [U]\rightarrow [K]$ as before.
Suppose $x_i$ is the only occurrence of elements in $S_i$ in window $i$.
The algorithm remembers $v=f(x_i)$, and starts to check \textbf{in parallel} if the heavy hitter has hash value $v$, and is in $S_i$.
That is, the algorithm initiates a $\checkfn$ procedure, while at the same time continues to sample $S_{i+1},S_{i+2}$, and so on, and may continue to initiate more $\checkfn$ procedures from future windows.
As we will see below, each $\checkfn$ terminates after running for $O(1)$ windows in expectation, and only $O(1)$ $\checkfn$ will be running at the same time in parallel (enforced by the algorithm).

Now consider one $\checkfn$ initiated from window $i$ due to sampling $x_i\in S_i$.
If $x_i\neq x^*$, then for each window that \emph{does not contain} $x_i$, there is only $O(1/K)$ elements sampled in $S_i$ \emph{and} with hash value $v$ in expectation.
Thus, within $O(f_{x_i})$ windows, we will be able to notice that there are many windows that do not contain $x_i$.
That is, it takes $O(f_{x_i})$ windows to figure out that $x_i$ is not the super heavy hitter, and this $\checkfn$ can return NO.
To see that we will not have too many $\checkfn$ procedures running at the same time, let us consider the expected total running time of all $\checkfn$ initiated throughout the stream.
For each element $x$, it may appear in at most $f_x$ different windows; in each window where it appears, there is a probability $q=\Theta(1/\sqrt{n})$ of sampling it; for each time it gets sampled, it takes at most $O(f_x)$ more windows for $\checkfn$ to figure out it is not the super heavy hitter (or $x=x^*$).
Hence, by linearity of expectation, the total running time of all $\checkfn$ is
\[
    \sum_x f_x\cdot q\cdot O(f_x)=O(\sqrt{n})
\]
windows.
That is, an average window only has $O(1)$ $\checkfn$ running.
We can also get a worst-case bound by restricting the total number of $\checkfn$ that can be executed at the same time in the algorithm, i.e., we do not start new $\checkfn$ procedures if it already hits the limit.
It turns out that this restriction does not impact the success probability by more than a constant.

\subsection{$K$-super heavy hitter}
To be more concrete on the $\checkfn$ procedure, consider a $\checkfn$ initiated due to sampling $x_i\in S_i$.
Note that we cannot afford to remember the identity of $x_i$ in order to check if it appears in every window.
Moreover, we cannot even remember the window index ``$i$'' (recall we have $O(\log n/\log\log n)$ hash functions, and each of them can occupy only $O(\log\log n)$ bits).
If we actually could remember the index $i$, then it suffices to verify that every window contains an element with hash value $v=f(x_i)$ (which we can remember) and is contained in $S_i$ by the argument in the last subsection, since
\begin{itemize}
    \item this holds for the heavy hitter $x^*$;
    \item if $x_i\neq x^*$, we will be able to notice it with high probability within $O(f_x)$ windows.
\end{itemize}

The main challenge here is that we cannot remember in which window we sampled the element.
Our solution first breaks the $\checkfn$ into $O(\log n)$ iterations with geometrically increasing numbers of windows.
In iteration $l$, we aim to verify $2^l$ consecutive windows.
We always start the iteration from the next window $j$ such that $j$ is a multiple of $2^l$.\footnote{This allows us to not remember when we started the iteration by only remembering the current window number. We can afford to remember $l$ (taking $O(\log l)=O(\log\log n)$ bits), but not an index within a $2^l$ range. The current window number takes $O(\log n)$ bits, but can be shared across all hash functions.}
In this case, from the fact that we are currently executing iteration $l$, and the index of the current window, we will be able to derive a range which $i$ may be in.
One can verify that this is a range $I$ of $2^l$ windows.

Thus, Let $S'$ be the union of $S_k$ for all $k\in I$.
Instead of verifying elements with hash value $v$ and is contained in $S_i$, we verify elements with hash value $v$ and is contained in $S'$.
By doing a standard expectation calculation, we can show that if $x_i$ is the heavy hitter $x^*$, then within the $2^l$ windows we execute iteration $l$, $x_i$ is a $K$-super heavy hitter among all elements with hash value $v$ and is contained in $S'$, i.e.,  within the $2^l$ windows, the frequency of $x_i$ squared is at least $K$ times more than the sum of that over all other elements.

Finally, for $K$-super heavy hitter, we apply another sample-and-check algorithm, which uses $O(\log\log n)$ bits most of the time, and uses $O(\log n)$ bits only $1/\poly\log n$-fraction of the time (see~\cref{Algorithm 1}).
By strictly enforcing only $O(1)$ parallel instances of the algorithm may use $O(\log n)$ bits at the same time, we obtain the space bound.


At last, we note that the algorithm generalizes to $\epsilon$-heavy hitters in a natural way by changing the window size and sampling probability to reflect the frequency of $\epsilon$-heavy hitters.
We will also have more hash functions (as we allow more space) so that each heavy hitter is sampled by sufficiently many hash functions.
\section{Preliminaries}
In this section, we formally define random order streams and partially random order streams, and state some useful concentration bounds.

\subsection{Problem definitions}

\begin{definition}[Adversarial streams]\label{def:adversarial streams}
An adversarial stream is defined by a \emph{multiset} $S$ of $n$ elements from the universe $U$, and a bijective map $\Pi:S\rightarrow [n]$ that maps every element in $S$ to a unique position in the stream. 
\end{definition}
For example, the multiset $S=\{a,a,b\}$ and the bijective map $\Pi\{a\rightarrow 1, a\rightarrow 3, b\rightarrow 2\}$, correspond to the stream $a,b,a$.

\begin{definition}[Random order streams]\label{def:random order streams}
A random order stream is defined by a multiset $S$ of $n$ elements from the universe $U$, and a \emph{uniformly random} bijective map $\Pi:S\rightarrow [n]$.
\end{definition}

In the following, we define partially random order streams for every fixed $\epsilon > 0$.

\begin{definition}[Partially random order streams] \label{def:partially random order streams}
A partially random order stream is defined by a multiset $S$ of $n$ elements from the universe $U$, and a bijective map $\Pi:S\rightarrow [n]$ such that the restricted map $\Pi|_{S^\epsilon}: S^{\epsilon}\rightarrow [n]$ is uniformly random, where ${S^\epsilon}$ is the subset of $\ell_2$ $\epsilon$-heavy hitters in $S$.
\end{definition}
\subsection{Concentration bounds}

We now recall some standard concentration bounds like Chebyshev's inequality, Chernoff bound, and Hoeffding's inequality. We also state and prove a (somewhat standard) Azuma-Hoeffding style inequality. Similar inequalities have been proved in other works in the literature (for example, see Lemma 2.5 in \cite{KK19}). 

\begin{theorem}[Chebyshev's inequality]\label{thm:Chebyshev}
    Let $X$ be a random variable with variance $\sigma^2$. For any $a>0$, we have
    \[\Pr[|X-\Exp[X]|> a] \le \frac{\sigma^2}{a^2}\, .\]
\end{theorem}

\begin{theorem}[Chernoff bound]\label{thm:chernoff}
Let $X_1,\dots,X_n$ be independent random variables taking values in $\{0,1\}$. Let $X = X_1 + \cdots + X_n$ and $\mu = \Exp[X]$. For every $\delta>0$, we have
\[
\Pr[X\ge(1+\delta)\mu] \le \left(\frac{e^\delta}{(1+\delta)^{1+\delta}}\right)^\mu \, .
\]
\end{theorem}
\begin{theorem}[Hoeffding's inequality]\label{thm:hoeffding's}
Let $X_1,\dots,X_n$ be independent bounded random variables such that $a_i\le X_i\le b_i$. Consider the sum of these random variables, $S_n = X_1 + \cdots + X_n$. Then Hoeffding's theorem states that for all $t>0$,
\[\Pr[|S_n-\Exp[S_n]|\ge t]\le 2\exp\left(-\frac{2t^2}{\sum_{i=1}^n(b_i-a_i)^2}\right) \, .\]
\end{theorem}

\begin{theorem}[Azuma-Hoeffding style concentration inequality]\label{thm:azuma-hoeffing}
Let $X=\sum_{i\in N} X_i$ where $X_i$ are Bernoulli random variables such that for every $k\in [N]$, $\Exp[X_k\mid X_1,\dots,X_{k-1}] \ge p$ for some $p\in (0,1)$. Let $\mu = Np$. Then 
\begin{itemize}
    \item for $0<\Delta \le \mu/2$, we have \[\Pr[X\le \mu-\Delta] \le \exp\left(-\frac{\Delta^2}{4(\mu -\Delta)}\right)\, , \]
    \item for $\Delta>\mu/2$, we have \[\Pr[X\le \mu-\Delta] \le \exp\left(\frac{-3\Delta+\mu}{4}\right)\, . \]
\end{itemize}
\end{theorem}

\begin{proof}
First, observe that for random variables $Z,Y$ if $\Exp[Z|Y]\le a$, then $\Exp[ZY]\le a \Exp[Y] $. For $u>0$, consider the random variables $e^{-uX_1},\dots,e^{-uX_N}$. We have for every $k\in [N]$, $\Exp[e^{-uX_k}\mid e^{-uX_1},\dots,e^{-uX_{k-1}}] \le (1-p)+p e^{-u}$. Now applying our previous observation, we have
\[\Exp[e^{-uX}] \le (1-p+pe^{-u}) \Exp\left[\prod_{i=1}^{N-1} e^{-uX_i}\right] \le \cdots \le (1-p+pe^{-u})^N \, . \]
Applying Markov's inequality, we have $\Pr[e^{-uX} \ge e^{-u(\mu - \Delta)}] \le \frac{(1-p+pe^{-u})^N}{e^{-u(\mu - \Delta)}}$.
Therefore,
\[\Pr[X\le \mu - \Delta] \le \min_{u>0} \frac{(1-p+pe^{-u})^N}{e^{-u(\mu - \Delta)}} \, .\]
Substituting $v=1-e^{-u}$ gives us
\[\Pr[X\le \mu - \Delta] \le \min_{0<v<1} \frac{(1-pv)^N}{(1-v)^{\mu - \Delta}} \, .\] 
Using the fact that $e^{-x-x^2}\le 1-x\le e^{-x}$ for $0\le x\le 1/2$, we can bound $\Pr[X\le \mu - \Delta]$ as
\[\Pr[X\le \mu - \Delta] \le \min_{0<v\le 1/2} \exp\left(-\mu v + (v+v^2)(\mu -\Delta)\right) = \min_{0<v\le 1/2} \exp\left(v^2(\mu -\Delta) - \Delta v \right) \, .\]
$g(v) = v^2(\mu -\Delta) - \Delta v$ is an upward-facing parabola that attains minima at $v=\frac{\Delta}{2(\mu-\Delta)}$. When $\Delta \le \mu/2$, $\frac{\Delta}{2(\mu-\Delta)}\le 1/2$, and hence, we can bound \[\Pr[X\le \mu - \Delta] \le \exp\left(-\frac{\Delta^2}{4(\mu -\Delta)}\right) \, .\] For $\Delta > \mu/2$, the minimum value for $v\in(0,1/2]$ is achieved at $v=1/2$. Substituting $v=1/2$ in $g(v)$, we get
\[\Pr[X\le \mu-\Delta] \le \exp\left(\frac{-3\Delta+\mu}{4}\right)\, .\]
\end{proof}

We also prove the following concentration bound for sampling without replacement. 
\begin{theorem}[Sampling without replacement]\label{thm:sampling without replacement}
Consider a bag containing $n$ balls numbered $1,\dots,n$. Let $k<n$ balls be sampled at random without replacement from this bag. Let $\ell\in [n]$ and $i\in[0,n-\ell]$. Let $X_{\ell,i}$ denote the number of balls that are sampled whose numbers lie in the range $i+1$ to $i+\ell$. Then,
\[
\Pr\left[X_{\ell,i} < 2\ell k/3n\right] \le 2\exp(-\ell k/18 n) \, .
\]
    
\end{theorem}

Before we prove \cref{thm:sampling without replacement}, we state a few facts from probability theory that are useful to prove this result.
A collection of random variables are said to be \emph{negatively associated} if they satisfy the following conditions.
\begin{definition}[Negatively associated random variables]\label{def:NA random variables}
The random variables $\mathbf{X}=(X_1,\dots,X_n)$ are negatively associated if for every disjoint index sets $I,J\subseteq [n]$,
\[
\Exp[f(X_i,i\in I)g(X_j,j\in J] \le \Exp[f(X_i,i\in I)] \Exp[g(X_j,j\in J] \, ,
\] for all functions $f:\mathbb{R}^{|I|}\rightarrow \mathbb{R}$ and $g:\mathbb{R}^{|J|}\rightarrow \mathbb{R}$ that are either both non-increasing or both non-decreasing.
\end{definition}
A theorem due to Dubhashi and Ranjan \cite[Proposition 5]{BallsAndBins} states that Chernoff-Hoeffding bounds can be applied on negatively associated random variables. In particular, we have the following.

\begin{theorem}[Chernoff bound for negatively associated random variables]\label{thm:NA chernoff}
Let $\mathbf{X}=(X_1,\dots,X_n)$ be negatively associated identically distributed Bernoulli random variables. Consider the sum of these random variables, $S_n = X_1 + \cdots + X_n$. Let $\mu = \Exp[S_n]$. Then for all $0\le \delta\le 1$,
\[\Pr[S_n\le (1-\delta)\mu]\le 2\exp\left(-\delta^2\mu /2\right) \, .\]
\end{theorem}

\begin{theorem}[Hoeffding's inequality for negatively associated random variables]\label{thm: NA Hoeffding}
Let $X_1,\dots,X_n$ be negatively associated bounded random variables such that $a_i\le X_i\le b_i$. Consider the sum of these random variables, $S_n = X_1 + \cdots + X_n$. Then Hoeffding's theorem states that for all $t>0$,
\[\Pr[|S_n-\Exp[S_n]|\ge t]\le 2\exp\left(-\frac{2t^2}{\sum_{i=1}^n(b_i-a_i)^2}\right) \, .\]    
\end{theorem}

The following proposition from \cite[Proposition 7]{BallsAndBins} states some useful properties of negatively associated random variables.
\begin{proposition}\label{prop:NA}
Let $\mathbf{X}=(X_1,\dots,X_n)$ be negatively associated random variables.
\begin{enumerate}
    \item For any index set $I\subseteq [n]$, the random variables $(X_i,i\in I)$ are also negatively associated.\label{prop:NA subset}
    \item If the random variables $\mathbf{Y}=(Y_1,\dots,Y_m)$ are negatively associated and are mutually independent from $\mathbf{X}$, then the augmented vector $(\mathbf{X},\mathbf{Y})$ of random variables are also negatively associated.\label{prop: NA augmented}
    \item Let $I_1,\dots,I_k\subseteq[n]$ be disjoint index sets, for some positive integer $r$. For $j\in [k]$, let $h_j:\mathbb{R}^{|I_j|}\rightarrow\mathbb{R}$ be functions that are all non-decreasing or all non-increasing, and define $Y_j:=h_j(X_i,i\in I_j)$. Then the vector of random variables $\mathbf{Y}=(Y_1,\dots,Y_k)$ are also negatively associated. In other words, non-decreasing (or non-increasing) functions of disjoint subsets of negatively associated variables are also negatively associated.\label{prop: NA composition}
\end{enumerate}
\end{proposition}

We are now ready to prove \cref{thm:sampling without replacement}.

\begin{proof}
For $j\in [n]$, let $B_j$ be the indicator random variable for the event that ball numbered $j$ is sampled. Dubhashi, Priebe, and Ranjan proved in \cite[Corollary 11]{Dubhashi_Priebe_Ranjan_1996} that the random variables $(B_j:j\in [n])$ are negatively associated. It follows from \cref{prop:NA subset} in \cref{prop:NA} that $(B_{j}:i+1\le j \le i+\ell)$ are also negatively associated. Since $X_{\ell,i} = \sum_{j=i+1}^{i+\ell} B_j$, $\Exp[X_{\ell,j}]=\ell k/n$, and the variables $B_j$ are identically distributed, the desired bound follows from applying \cref{thm:NA chernoff}. 
\end{proof}
\section{Streaming algorithm for detecting $\ell_2$ \emph{super} heavy hitter}\label{sec:KSHH}

In this section, we describe our subroutine for detecting super-heavy hitters and prove some lemmas that will be useful in the analysis of our main heavy-hitter algorithms in \cref{sec:SHH} and \cref{sec:epsilon-heavy}.

Let $\s$ be an input stream divided into $L$ windows, each of length $W$. Let $m = L W$ denote the total length of the stream. Let $U$ be a universe of size at most $\poly\, n$\footnote{Every element in the universe can be represented using $O(\log n)$ bits.}.
$K = (\log n)^{100}$ is a fixed parameter.
We say that an element $x$ is a $K$-super heavy hitter if $f_x^2\geq K\cdot \sum_{y\neq x} f_y^2$, or equivalently it is a $\left(1-\frac{1}{K+1}\right)$-heavy hitter. We are given a subset of the universe $V\subset U$, and we focus on the substream consisting of elements in $V$. The positive case is that there exists a $K$-super heavy hitter $x^*\in V$ (with respect to the substream restricted to $V$) of frequency $L$ and \emph{evenly distributed} in the original stream. In this section, we develop algorithms that return YES if we are in the positive case, with high probability, return NO if we are far from it (i.e., every element in $V$ has frequency at most $K^{-1/8} L $), and with further guarantees for the intermediate cases.
In the following section, we first give a streaming algorithm for the case when $L\ge K^{3/16}$, and in the end, we handle the case when $L< K^{3/16}$ separately.
\subsection{$L\ge K^{3/16}$}\label{sec:largeL}
\begin{algorithm}[H]
\caption{Streaming algorithm for evenly distributed $\ell_2$ super-heavy hitters ($K$-SHH)}
\label{Algorithm 1}
\begin{algorithmic}[1]
\Input a stream $S$ divided into $L$ windows, each of length $W$ ($m = L W$ denotes the total length of the stream), parameter $K = (\log n)^{100}$ such that $L\ge K^{3/16}$, and a subset $V$ of the universe $U$ of size $\poly\, n$
\State set $T=100 \log m$ and $p=(2/5)T/L$.
\State initialize $\mathsf{counter}=0$.
\State subdivide the stream into $T$ consecutive segments, each containing $L/T$ windows.
\For{each segment}
\State process the segment by sampling each element of the stream that is in $V$ independently with probability $p$ until an element is sampled or there are no more elements left in the current segment
\State proceed to the following step if the current window is not \emph{blocked}. Else, continue the sampling process\label{algo_KSHH_line-block}
\If{an element is sampled}
\State set $y$ to be this sampled element and $\mathsf{t}=1$.
\Else
\State go to the next segment.
\EndIf
\If{the total number of remaining windows within the segment is less than $L/K^{1/8}$}
\State go to the next segment.
\ElsIf{$y$ occurs in each of the next $L/K^{1/8}$ windows} \label{alg1_check_line-14}
\State increment $\mathsf{counter}$ by $1$, and go to the next segment.
\Else
\State return to Line $5$.
\EndIf
\EndFor
\Output return YES if $\mathsf{counter}>T/4$ and NO otherwise.
\end{algorithmic}
\end{algorithm}

Consider the following (hypothetical) algorithm $\mathcal{A}$\footnote{We define $\mathcal{A}$ purely for the sake of the analysis.} which is a ``space-unconstrained'' modification of \cref{Algorithm 1}: The execution of $\mathcal{A}$ is same as that of \cref{Algorithm 1} except that in Line 5, $\mathcal{A}$ continuously samples elements in the segment. In particular, even when an element is already sampled in Line 5, $\mathcal{A}$ does not pause sampling elements in the current segment. For all the sampled elements, the Lines 7-18 are executed concurrently with the sampling process in Line 5. In addition, $\mathcal{A}$ and \cref{Algorithm 1} share the same random coin tosses\footnote{For the analysis sake, we can imagine that \cref{Algorithm 1} also tosses a random coin whenever it sees an element of the stream that is in $V$, but unlike $\mathcal{A}$ discards the outcomes of these random coins when there is already a sampled element that is being checked.}. The hypothetical algorithm has no space constraints and hence Line 6 is never executed. Finally, since Lines 7-18 are executed independently for all the sampled elements, we find that compared to \cref{Algorithm 1}, it is easier to analyze $\mathcal{A}$ directly. In the following lemmas, we first analyze $\mathcal{A}$ and then use it to analyze \cref{Algorithm 1}.

In the following, we analyze the space used by $\mathcal{A}$ in terms of the number of windows in which it uses $\log n$ bits.

\begin{lemma}\label{lem_KSHH_expectation}
The expected total number of windows in which $\mathcal{A}$ uses $\log n$ bits is at most
\[
\left(\sum_{x:f_x\leq L/K^{1/8}} f_x^2\right) \cdot O(T/L) + \left(\sum_{x:f_x> L/K^{1/8}} f_x \right) \cdot O(T/K^{1/8}) \, .
\]
\end{lemma}

\begin{proof}
Note that $\mathcal{A}$ uses $\log n$ bits only when executing Lines 7-18 and checking for the occurrences of the sampled element.
    For any element $x$ in the universe, the expected number of times it is sampled is $f_x \cdot O(T/L)$. Each time $x$ is sampled, the Lines 7-18 are executed in at most $\max\{L/K^{1/8}, f_x\}$ windows.
\end{proof}

\begin{lemma}[Concentration lemma for the high-frequency elements]\label{lem:high_freq_conc}
Suppose that $\sum_{x:f_x> L/K^{1/8}} f_x \le L \gamma$. With probability at least $1-1/\poly\, m$, the total number of windows in which $\mathcal{A}$ uses $\log n$ bits after sampling a high-frequency (frequency larger than $L/K^{1/8}$) element is at most $L \gamma (\log m)^{2}/K^{1/8}$.
\end{lemma}

\begin{proof}
It suffices to show that the total number of high-frequency elements sampled by $\mathcal{A}$ is at most $\gamma (\log m)^{2}$ with probability at least $1-1/\poly\, m$. Let $R=\sum_{x: f_x > L/K^{1/8}} f_x$. For $i\in [R]$, let $X_i$ be an indicator random variable for the event that the $i$-th high-frequency element in the stream is sampled by $\mathcal{A}$. We know that each element is sampled independently with probability $p$. Let $X=\sum_{i\in [R]} X_i$. By assumption, we have $R \le L \gamma$. Therefore, $\Exp[X]\le \gamma \log m$. Applying Chernoff bound (see \cref{thm:chernoff}), we conclude that $\Pr[X > \gamma (\log m)^{2}] \le 1/\poly\, m$. 
\end{proof}

\begin{lemma}[Concentration lemma for the low-frequency elements]\label{lem:low_freq_conc}
   Suppose that $\sum_{x:f_x\leq L/K^{1/8}} f_x^2\le L^2/(\log m)^c$. With probability at least $1-1/\poly\, m$, the total number of windows in which $\mathcal{A}$ uses $\log n$ bits after sampling a low-frequency (frequency at most $L/K^{1/8}$) element is at most $L/(\log m)^{c/2-1}$.
\end{lemma}
\begin{proof}
    For an element $a$, let $f_a$ denote its frequency and let $X_a$ be a random variable that denotes the number of windows in algorithm $\mathcal{A}$ in which either $a$ is sampled or the window is used to check for the occurrence of $a$. Observe that $X_a$ is a bounded random variable, i.e., $X_a\le 2f_a$. We have $\Exp[X_a] \le p \cdot f_a^2$ since in every window that $a$ occurs, it is sampled with probability at most $p$ and then at most the following $f_a$ windows are used to check for the occurrence of $a$. Observe also that by definition of $\mathcal{A}$, the $X_a$'s are independent random variables. Let $N_\mathcal{A}$ denote the random variable for the total number of windows in which either a low-frequency element is sampled or checked by $\mathcal{A}$. We have $N_\mathcal{A}\le \sum_{a: f_a \le L/k^{1/8}} X_a$ and by linearity of expectation, $\Exp[N_\mathcal{A}]\le p \cdot \sum_{a:f_a\leq L/K^{1/8}} f_a^2 \le O(L/(\log m)^{c-1})$. By applying Hoeffding's inequality (see \cref{thm:hoeffding's}), we have
\[\Pr\left[N_\mathcal{A} \ge L/(\log m)^{c/2-1}\right] \le 2\exp\left(-2 \frac{ L^2/(\log m)^{ c-2} }{4 \sum_{a:f_a\leq L/K^{1/8}} f_a^2 }\right) = 2\exp\left(-\Theta(\log m)^{2}\right) \, .\]    
\end{proof}

In the following, we analyze the YES and the NO cases, and prove guarantees for the correctness of the algorithm.
\begin{lemma}[YES case]\label{lem:Yes_case}
Let $x^*\in V$ occur at least once in every window and let $\sum_{y\ne x^*\in V} f_y^2 \le L^2 \log^5 n/K$. In addition, suppose there is a set of at most $L/c$ \emph{blocked} windows ($c>=168$) where \cref{Algorithm 1} is not allowed to sample any new elements. Then \cref{Algorithm 1} returns YES with probability at least $1-\exp(-\Theta(T))$.
\end{lemma}
\begin{proof}
We say that a window is bad if an element $y\ne x^*$ is sampled in it or if it is used in Lines 7-18 to check for the occurrence of $y\ne x^*$. We first obtain a bound on the total number of bad windows in \cref{Algorithm 1}.
Observe that for any fixed outcomes of the random coin tosses, the number of bad windows in \cref{Algorithm 1} is upper bounded by the number of bad windows in $\mathcal{A}$.
For an element $a\ne x^*\in V$, let $f_a$ denote its frequency and let $X_a$ be a random variable that denotes the number of windows in algorithm $\mathcal{A}$ in which either $a$ is sampled or the window is used to check for the occurrence of $a$. Observe that $X_a$ is a bounded random variable, i.e., $X_a\le 2f_a$. We have $\Exp[X_a] \le p \cdot (f_a)^2$ since in every window that $a$ occurs, it is sampled with probability at most $p$ and then at most the following $f_a$ windows are used to check for the occurrence of $a$. Observe also that by definition of $\mathcal{A}$, the $X_a$'s are independent random variables. Let $N_\mathcal{A}$ denote the random variable for the total number of bad windows in $\mathcal{A}$. We have $N_\mathcal{A}\le \sum_{a\ne x^*\in V} X_a$ and by linearity of expectation, $\Exp[N_\mathcal{A}]\le p \cdot F \le L/(\log m)^{2}$, where the last inequality follows from the fact that $F\le L^2/K^{0.99}$.
By applying Hoeffding's inequality (see \cref{thm:hoeffding's}), we have
\[\Pr\left[N_\mathcal{A} \ge L/(\log m)\right] \le 2\exp\left(-\frac{(L/(\log m))^2}{4\sum_{a\ne x^*\in V}f_a^2}\right) =  \exp\left(-\Theta(\log m)\right)\, .\]

Let $N$ denote the random variable for the total number of bad windows and blocked windows in \cref{Algorithm 1}. Since the number of bad windows in \cref{Algorithm 1} is always less than $N_\mathcal{A}$ and the number of blocked windows is at most $L/c$, we get that 
\[\Pr\left[N \ge 2L/c\right] \le \exp\left(-\Theta(\log m)\right)\, .\]
Conditioned on $N<2L/c$, we will now upper bound the probability that at the end of \cref{Algorithm 1}, the $\mathsf{counter}$ value is less than $T/4$.
Let $Z = L/T - L/K^{1/8} \ge 6L/7T$. For segment $j$, let $I_j$ be the indicator random variable for the event that $x^*$ is sampled within the first $Z$ windows in the segment. Observe that the $\mathsf{counter}$ is incremented in Line 15 whenever $I_j=1$ since $x^*$ appears in every window. Hence, $\mathsf{counter} \ge \sum_{j\in [T]} I_j$. Let $N_j$ denote the total number of bad and blocked windows within the first $U$ windows in segment $j$. We have \[\Exp\left[I_j\bigg| N_1,\dots,N_T,N<2L/c\right]\ge 1 - (1-p)^{Z-N_j} \ge 1 - (1-p(Z-N_j)+(p(Z-N_j))^2/2) \ge 3p(Z-N_j)/4\, ,\] since $pZ\le 1/2$.
Hence, \[\Exp\left[\mathsf{counter}\bigg| N_1,\dots,N_T,N< 2L/c\right]\ge \frac{3pZT}{4} - (3p/4)N \ge \frac{3pZT}{4} - \frac{3pL}{2 c} \ge \frac{T}{4} \cdot \frac{71}{70} \, . \]
Observe that conditioned on $N_1,\dots,N_T$, the $I_j$'s are independent.
Hence, by applying Hoeffding's inequality, we have
\[\Pr\left[\mathsf{counter} < T/4 \bigg| N_1,\dots,N_T,N < 2L/c \right] \le \exp\left(-\Theta(T)\right)\, .\]  
\end{proof}

\begin{lemma}[NO case(s)]\label{lem:KSHH_no}
If there is no element with frequency $L/K^{1/8}$ or
if there is only one element with frequency at least $L/K^{1/8}$ but at most $L/2$, then \cref{Algorithm 1} returns NO with probability at least $1-\exp(-\Theta(T))$.
\end{lemma}
\begin{proof}
The $\mathsf{counter}$ is incremented only if the sampled element has frequency at least $L/K^{1/8}$. Hence, if there is no element with frequency at least $L/K^{1/8}$, the algorithm always returns NO.

In the second case where there is only one element $x$ with frequency at least $L/k^{1/8}$, the element must be sampled at least $T/4$ times for \cref{Algorithm 1} to return YES. We will show that with probability at least $1-\exp(-\Theta(T))$, $x$ is sampled less than $T/4$ times in \cref{Algorithm 1}. 
Consider algorithm $\mathcal{A}$.
Since whenever an element is sampled in \cref{Algorithm 1}, it is also sampled in $\mathcal{A}$, it suffices to show that with probability at least $1-\exp(-\Theta(T))$, $x$ is sampled less than $T/4$ times in $\mathcal{A}$.
Let $X_i$ be the indicator random variable for the event that $x$ is sampled by $\mathcal{A}$ at its $i$-th occurrence in the stream. We are interested in the sum $X=\sum_{i=1}^{f_x} X_i$, where $f_x$ is the frequency of $x$. Each $X_i$ is an independent Bernoulli random variable with $\Exp[X_i]=p$. By linearity of expectation, we have
\[\Exp[X] = f\cdot p \le T/6\, .\]
By applying the Chernoff bound, we have
\[\Pr[X\ge T/4]\le \exp(-\Theta(T))\, .\]    
\end{proof}

We now address the case when the total number of windows is too small, i.e., $L < K^{3/16}$. In this case, we simply hash every element in $V$ using a uniform $O(\log K)$-bit hash function and return YES if there are at least $L$ elements in the substream restricted to $V$ and all their hash values are equal, else return NO. We will now argue that \cref{lem_KSHH_expectation,lem:high_freq_conc,lem:low_freq_conc,lem:Yes_case,lem:KSHH_no} hold in this case as well. Observe that \cref{lem_KSHH_expectation,lem:high_freq_conc,lem:low_freq_conc} hold trivially because the algorithm never uses more than $O(\log K)$ bits. In the YES case corresponding to \cref{lem:Yes_case}, the condition that $\sum_{y\ne x^*\in V} f_y^2 \le L^2 \log^5 n/K$ implies that there is no element other than $x^*$ in the substream restricted to $V$. Therefore, the algorithm always outputs YES in this case. On the other hand, consider the NO case. If the substream restricted to $V$ has fewer than $L$ elements, then the algorithm always outputs NO. If the substream has at least $L$ elements, then the conditions in the NO case imply that there are at least $K^{1/8}/2$ distinct elements in the substream. The probability that they all have the same hash value is at most $1/\poly\, n$.




\section{Streaming algorithm for detecting $\ell_2$ heavy hitter}\label{sec:SHH}
In this section, we present the algorithm for finding an $\ell_2$ heavy hitter with high probability using $O(\log n)$-bits of space, assuming the \emph{heavy hitter} is evenly distributed in the stream.
Here, we assume that the length of the stream is $n$, the second moment of the frequency vector is $\approx C^2 n$, and we assume access to free randomness.

The algorithm is as follows (\cref{alg:ell2}, \cref{alg:check}).
At a high level, we sample $O(\log n/\log\log n)$ independent hash functions with $O(\log\log n)$-bit hash values.
The goal is to compute for (most of) the hash functions, what is the hash value of the heavy hitter, if exists.
To this end, we partition the stream into windows of size approximately $\sqrt{n}$ so that each window has on average on occurrence of the heavy hitter.
Then we sample a random set $S_i\subseteq U$ for each window, each element is contained in $S_i$ with probability $\approx 1/\sqrt{n}$ independently.
Note that, with constant probability, some window will sample the heavy hitter, \emph{and} no other element appearing in the window is sampled.
We hope this happens, and for the sampled elements, we remember its hash value, and run a $\checkfn$ procedure to check if this is the hash value of the heavy hitter.

For the $\checkfn$ procedure, the high level idea is to gradually increase the length of the interval.
In the $l$-th iteration, we will examine $2^l$ consecutive windows, whether there is a heavy hitter with the given hash value.
As argued in the overview, this is done by invoking the $K$-super-heavy-hitter subroutine from the previous section.


\begin{algorithm}[H]\caption{Algorithm for $\ell_2$ heavy hitters}\label{alg:ell2}
\begin{algorithmic}[1]
\Input a stream of length $n$ of elements from set $U$ of size $\poly\, n$ such that some element $x^*$ appears at least $C\sqrt{n}$ times
\State \emph{assumptions: dividing the stream into windows of length $W:=\sqrt{n}/C$, there is at least one occurrence of $x^*$ in each window; $\ell_2^2$ of the frequency vector is at most $(1+2^{-8})C^2 n$}
\Output $x^*$
\State fix parameters $q=\frac{1}{C\sqrt{n}}$, $K=\log^{1000} n$ and $J=128\log U/\log K$
\State let $h^{(1)},\ldots,h^{(J)}$ be independent random hash functions $h^{(j)}: U\rightarrow [K]$
\State for each window $i$ and each $j\in[J]$, sample a random set $S^{(j)}_i\subseteq U$ such that each $x\in U$ is in $S^{(j)}_i$ with probability $q$ independently \emph{using free random bits}
\For{each window $i$}
\State for $j\in[J]$, if there is some $x\in S^{(j)}_i$ that appears in the window, record $(j, h^{(j)}(x))$ (if there is more than one such $x$ for some $j$, record the first one)
\For {each hash $(j, v)$ recorded at the end of the window}
\If {at most $100$ instances of $\checkfn$ are still running for this $j$}\label{alg_ell2_line-8}
\State start $\checkfn(j, v)$ \emph{in parallel}
\EndIf
\EndFor
\EndFor
\State For each $j$, record the first $(j, v)$ such that $\checkfn(j, v)$ returned YES
\State If there is an $x$ such that $h^{(j)}(x)=v$ for at least $4\log U/\log K$ of them, return $x$ (if there is more than one such $x$, return any)
\end{algorithmic}
\end{algorithm}

For every $l=0,\ldots,\maxl = \log(C\sqrt{n}/8)$, we divide the stream into intervals of $2^l$ windows, i.e., for $l=0$, each segment is a single window; for $l=1$, the first segment consists of the first two windows, then the second consists of the next two windows.
Each length-$2^l$ interval is subdivided into two length-$2^{l-1}$ intervals.
We now build a binary forest on this multi-level division: For each length-$2^{l-1}$ interval, its parent is the \emph{chronologically next} interval of length-$2^l$.\footnote{Note that there is more than one node with no parent, hence, it is a forest.}
The $\checkfn$ function runs $K$-SHH on these intervals in increasing levels, the $l$-th round will be on the parent interval of the $l-1$-th round.
We call this binary forest the \emph{execution forest} of $\checkfn$ (see \cref{fig:execution_forest} below).

\begin{figure}[h!]
    \centering
    \includegraphics[width=0.5\linewidth]{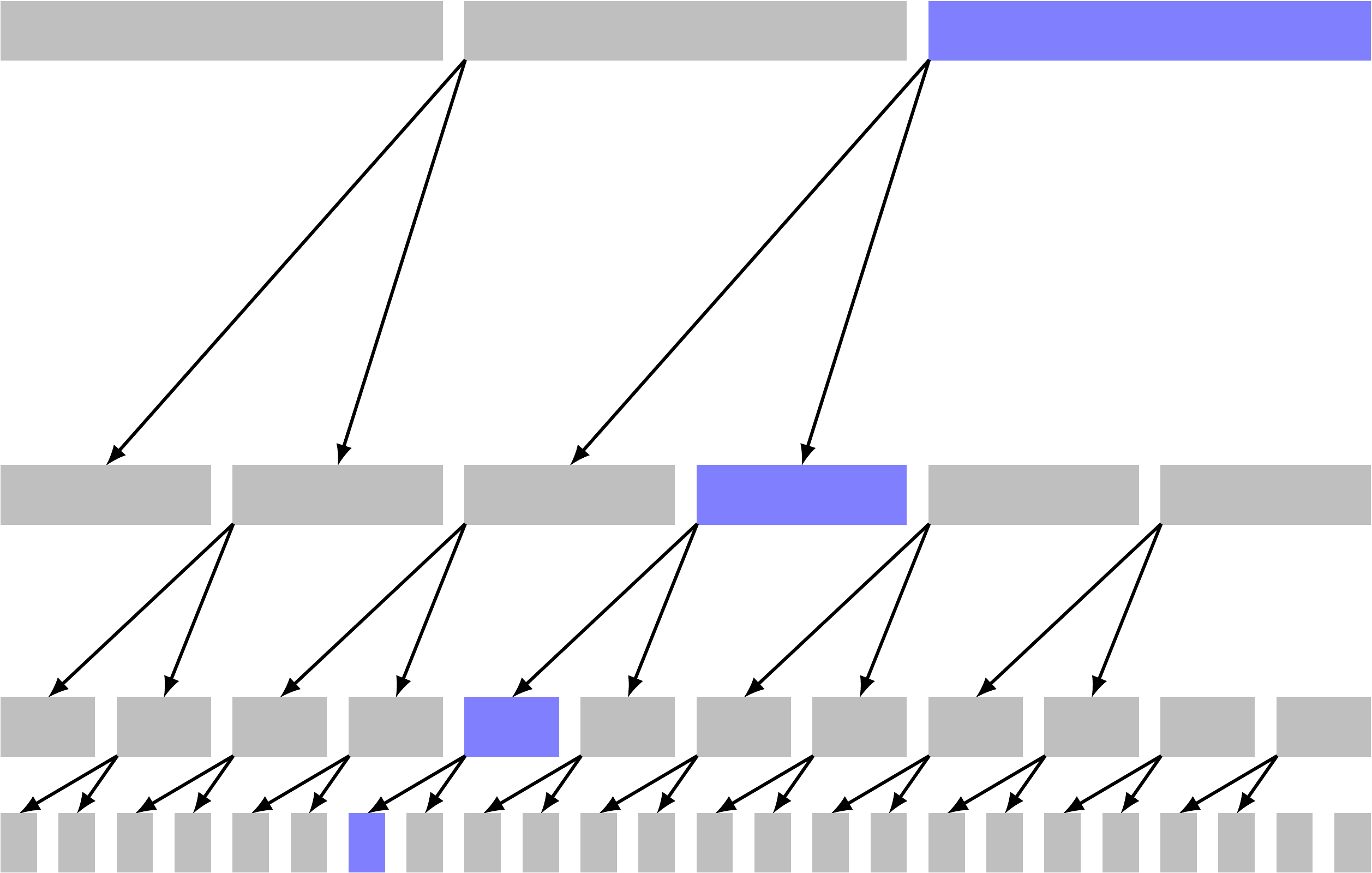}
    \caption{Execution forest of $\checkfn$. An example execution path of a $\checkfn$ function initiated at one of the leaf nodes is highlighted in blue.}
    \label{fig:execution_forest}
\end{figure}

\begin{algorithm}[H]\caption{$\checkfn(j,v)$ checks if $h^{(j)}(x^*)=v$  (when $x^*$ exists)}\label{alg:check}
\begin{algorithmic}[1]
    \State let $\maxl=\log(C\sqrt{n}/8)$
    \For {$l=0,\ldots,\maxl$}
	\State currently at the beginning of a level-$l$ node $I$ in the execution forest, let it be the beginning of the $i$-th window such that $2^l\mid i$; the current $\checkfn$ must be initiated from a window in the range $[\max\{0, i-2^{l+1}+1\}, i-2^l]$
        \State\label{line:unique_jv} if there is another instance of $\checkfn$ with the same $j$ and $v$ is about to start the same round $l$ (hence on the same $I$), then keep only one instance and terminate the others
        \State let $U_{I, j}\leftarrow \bigcup_{k\in [\max\{0, i-2^{l+1}+1\}, i-2^l]} S_k^{(j)}$ be the union of all $S_k^{(j)}$ of windows that may have initiated this $\checkfn$
	\State run $K$-SHH (\cref{Algorithm 1}) on $I$, restricted to elements in $U_{I, j}$ with hash value $v$ (i.e., set the parameter $V$ to $U_{I,j,v}$, $n$ is the length of the original stream, and we assume that the universe size is at most $\poly\,  n$, $W$ and $K$ are fixed parameters that do not vary across the different executions of $K$-SHH)
 $$U_{I, j, v}:=\{x: h^{(j)}(x)= v\}\cap U_{I, j},$$ 
 denote this instance of $K$-SHH by $A_{I,j,v}$;
 \textbf{only $K$-SHH instances $A_{I,j,v}$ with the same $j$ can use $O(\log n)$ bits at the same time (any $A_{I,j,v}$ that is using $O(\log n)$ \emph{blocks} all other $A_{I',j',v'}$ for $j'\neq j$, i.e., they do not proceed after \cref{algo_KSHH_line-block} in \cref{Algorithm 1})}\label{alg_check_line-6}
	\If {$K$-SHH returns NO}
            \State return NO
        \EndIf
        \State wait until the start of the next window $i$ such that $2^{l+1}\mid i$ 
    \EndFor
    \State return YES
\end{algorithmic}
\end{algorithm}

Like in the previous section, in order to analyze the algorithm, we need to analyze a \emph{hypothetical version} of it.
\begin{definition}
    The \emph{hypothetical version} of \cref{alg:ell2} may have any number of instances of $\checkfn$ running at the same time for each $j$ (no restriction in \cref{alg_ell2_line-8}). Also it uses the hypothetical versions of $K$-SHH, and the $K$-SHH subroutines may use $O(\log n)$ bits at the same time (no restriction in \cref{alg_check_line-6} in \cref{alg:check}).
\end{definition}

It is easy to verify that when using the same randomness, if the real algorithm has some $K$-SHH subroutine use $O(\log n)$ bits in a window, then the the same subroutine must also be using $O(\log n)$ bits in the same window in the hypothetical version.
Hence, the hypothetical version generally uses more space, but it has less dependence between different subroutines, making it easier to analyze.
In the following, most of the lemmas will be proving bounds on the hypothetical version directly, while they should also carry over to the actual algorithm.

\begin{lemma}\label{lem:bound_checkfn_time}
    In the hypothetical version of the algorithm, for each hash function $j\in[J]$, the expected number of windows with at most $100$ instances of $\checkfn$ running is at least $(3/4-o(1))\cdot n/W$. 
\end{lemma}
\begin{proof}
    Consider any node $I$ in the \emph{execution forest} of $\checkfn$ in level $l<\maxl$, fix a hash function index $j\in[J]$ and one possible hash value $v\in[K]$.
    Now let us analyze the probability that a $\checkfn$ procedure on $(j, v)$ passes the round-$l$ test on $I$, i.e., the probability that it advances to round $l+1$ on the parent node of $I$.
    Note that there can be at most \emph{one} $\checkfn$ on any particular $(j, v)$ that runs the round-$l$ test on $I$ due to \cref{line:unique_jv}.

    To this end, we will apply Lemma~\ref{lem:KSHH_no} and bound the probability of the necessary conditions for $\checkfn(j,v)$ to pass round $l$ on $I$.
    Recall that $K$-SHH is run on the substream restricted to the set of elements $U_{I,j,v}$.
    Lemma~\ref{lem:KSHH_no} ensures that if \emph{neither} of the following happens, then $K$-SHH must return NO except with $1/\poly\, n$ probability:
    \begin{enumerate}
        \item there exists an element in $U_{I,j,v}$ with frequency at least $2^{l-1}$ in $I$;
        \item there exist at least two elements in $U_{I,j,v}$ with frequency at least $2^l\cdot K^{-1/8}$ in $I$.
    \end{enumerate}
    Hence, it suffices to consider $\checkfn(j, v)$ of these two types, i.e., we say $\checkfn(j, v)$ has \emph{Type 1} with respect to $I$ if it passes round $l$ on $I$ and it satisfies the first condition above; similarly, $\checkfn(j, v)$ has \emph{Type 2} with respect to $I$ if it passes round $l$ on $I$ and satisfies the second condition.

    Let us first consider Type 1 $\checkfn$ procedures.
    For a fixed $I$, let $\alpha_I$ be the number of elements with frequency at least $2^{l-1}$ in $I$.
    For any given element, with probability $K^{-1}$, it has hash value $v$, and with probability at most $2^l\cdot q$, it belongs to one of $S_k^{(j)}$ that could have initiated $\checkfn$ to run on $I$ in round $l$.
    Each fixed element belongs to $U_{I,j,v}$ with probability at most $K^{-1}\cdot 2^l\cdot q$, hence, the probability that $\checkfn(j, v)$ has Type 1 with respect to $I$ is at most
    \[
        \alpha_I\cdot K^{-1}\cdot 2^l\cdot q.
    \]

    If $\checkfn(j,v)$ has Type 1 with respect to $I$, then it may wait for at most $2^l$ windows and will continue to execute on the parent node of $I$ for $2^{l+1}$ windows, a total of at most $3\times 2^l$ windows.
    Now for a fixed pair $(j, v)$, we will upper bound the \emph{total running time} (including the waiting time to get to the parent node) of $K$-SHH on all the \emph{parent nodes} of nodes where $\checkfn(j, v)$ is executed and has Type 1 with respect to that node.
    By linearity of expectation, the expected number of windows is at most
    \begin{equation}\label{eqn:runtime_alpha}
        \sum_{l\geq 0}\sum_{\textrm{$I$ in level }l} (\alpha_I\cdot K^{-1}\cdot 2^l\cdot q)\cdot 3\cdot 2^{l} = 3K^{-1} q \sum_{l\geq 0}\sum_{\textrm{$I$ in level }l} \alpha_I\cdot 4^l\, .
    \end{equation}
    The following claim lets us upper bound the RHS.
    \begin{claim}\label{clm:alpha_I}
    We have 
    \[
        \sum_{l\geq 0}\sum_{\textrm{$I$ in level }l} \alpha_I\cdot 4^l \leq (8+2^{-5})C^2 n.
    \]
    \end{claim}

    \begin{proof}[Proof of \cref{clm:alpha_I}]
        Recall that $\alpha_I$ is the number of elements with frequency at least $2^{l-1}$ in $I$.
    Consider any element $x$ with frequency $f_x$, its contribution to LHS in the claim is equal to
    \[
        \sum_{l\geq 0}\sum_{\textrm{$I$ in level }l: f_{x,I}\geq 2^{l-1}}  4^l,
    \]
    where $f_{x,I}$ is the frequency of $x$ in $I$.
    Observe that for any $I$ in level $l$, $\sum_{I'\subseteq I, I'\textrm{in level }l'<l}4^{l'}\leq 4^l$, where a node $I'$ is said to be contained in a node $I$ if the windows of $I'$ are a subset of the windows of $I$.  
    We say that a node $I$ at level $l$ is maximal with respect to $x$ if $f_{x,I} \ge 2^{l-1}$ and there exists no node $I'\supset I$ at level $l'>l$ such that $f_{x,I'} \ge 2^{l'-1}$.
    Hence, the above sum is bounded by twice the sum over all \emph{maximal} nodes, and these nodes must be disjoint, i.e., not share any windows.
    Therefore, we have
    \begin{align*}
        \sum_{l\geq 0}\sum_{I: \textrm{$I$ in level }l} \alpha_I\cdot 4^l &= \sum_x \sum_{l\geq 0}\sum_{\textrm{$I$ in level }l: f_{x,I}\geq 2^{l-1}}  4^l \\
        &\le 2\sum_x \sum_{l\geq 0}\sum_{\textrm{maximal $I$ in level }l \textrm{ w.r.t. } x}  4^l \\
        & \le 8\sum_x \sum_{l\geq 0}\sum_{\textrm{maximal $I$ in level }l \textrm{ w.r.t. } x}  f_{x,I}^2 \tag{$f_{x,I}\ge 2^{l-1}$}\\ 
        & \le 8\sum_x f_x^2 \tag{We are summing over maximal $I$ and $a^2 + b^2 \le (a+b)^2$ for $a,b\ge 0$}\\
        &\le (8+2^{-5})C^2 n. 
    \end{align*}
    The claim holds.
    \end{proof}
    
    Now, it follows from \cref{clm:alpha_I} that \eqref{eqn:runtime_alpha} is at most
    \[
        25C^2 n\cdot K^{-1}q.
    \]
    Taking the sum over all $v$, the expected total running time for a fixed hash function index $j\in [J]$ is at most
    \[
        25C^2 nq.
    \]
    Since there are $n/W=C^2 nq$ windows, by Markov's inequality, at most $1/4$ fraction of the windows have more than $100$ $\checkfn$s running due to passing the previous round with Type 1. 
    

    \bigskip

    To analyze Type-2 $\checkfn$ procedures, let $\beta_I$ be the number of elements with frequency at least $2^l\cdot K^{-1/8}$ in $I$.
    The probability that $\checkfn(j, v)$ has Type 2 with respect to $I$ is at most
    \[
        \binom{\beta_I}{2}\cdot (K^{-1}\cdot 2^l\cdot q)^2\leq 
        (\beta_I\cdot K^{-1}\cdot 2^l\cdot q)^2.
    \]
    We have the following claim, since all intervals $I$ in the same level are disjoint.
    \begin{claim}
    For each fixed $l$, we have 
    \[
        \sum_{\textrm{$I$ in level }l} \beta_I\cdot 4^l\cdot K^{-1/4} \leq (1+2^{-8})C^2 n.
    \]
    \end{claim}
    Since there are $n/(2^l W)$ nodes in level $l$, for an average node $I$, we have $$\beta_I\cdot 4^l\cdot K^{-1/4}\leq \frac{(1+2^{-8}) C^2 n}{n/(2^lW)}=(1+2^{-8})C^2 2^l W.$$ 
    That is, for an average node $I$,
    \[
        \beta_I\cdot K^{-1}\cdot 2^l\cdot q\leq (1+2^{-8})C^2 2^l W\cdot K^{-3/4}\cdot 2^{-l}\cdot q=O(K^{-3/4}).
    \]
    By Markov's inequality, except for at most $O(K^{-1/4}\log n)$ fraction of the nodes, we have 
    \begin{equation}\label{eqn:type2_bound}
        \beta_I\cdot K^{-1}\cdot 2^l\cdot q\leq K^{-1/2}\log^{-1} n.
    \end{equation}
    That is, there are at most $O(K^{-1/4}\log n)$-fraction of level $l+1$ nodes such that $\checkfn(j,v)$ can execute on them due to passing level $l$ on a node violating~\eqref{eqn:type2_bound}.
    Summing over all levels $l$, at most $K^{-1/4}\log^2 n=o(1)$-fraction of the windows can have such a $\checkfn$ running.
    

For a node $I$ satisfying~\eqref{eqn:type2_bound}, the probability that $\checkfn(j,v)$ has Type 2 with respect to $I$ is at most $K^{-1}\log^{-2} n$. Taking a union bound over all $v$, the probability that a $\checkfn$ corresponding to a fixed hash function $j\in [J]$ is of Type 2 with respect to $I$ is at most $\log^{-2} n$. If $\checkfn(j,v)$ has Type 2 with respect to $I$, then it will continue to execute on the parent node of $I$ for a total of at most $3\times 2^{l}$ windows. By linearity of expectation, the expected total number of windows executing such a check is at most
      \begin{equation}
        \sum_{l\geq 0}\sum_{\textrm{$I$ in level $l$ satisfying ~\eqref{eqn:type2_bound}}} \log^{-2} n \cdot 3 \cdot {2^\ell} \le 3 \log^{-2} n \sum_{l\geq 0}\sum_{\textrm{$I$ in level $l$}} {2^\ell} \le 3 \log^{-1} n  \cdot \left(n/W\right).
    \end{equation} By Markov's inequality, at most $o(1)$ fraction of all the windows have even a single instance of such a check running.

Combining the analysis for the two types, we prove the lemma.
\end{proof}

\begin{lemma}\label{lem:block_exp}
    Let $I$ be a node in the execution forest of $\checkfn$, and $I_1,I_2$ be its children.
    Suppose $I$ is in level $l$ such that $2^lW\leq n/\log^{10}n$, and $F_{I_1},F_{I_2},F_I\leq O(2^lW\log^4 n)$. 
    For a fixed $j\in [J]$ and hash function $h^{(j)}$, the expected total number of windows $A_{I,j,v}$ blocks for all $v$, \emph{conditioned on $h^{(j)}$}, is at most $O(2^l/\log^4 n)$.
\end{lemma}
\begin{proof}

    Recall that $\checkfn(j,v)$ passes round $l-1$ on $I_1$ if either
    \begin{itemize}
        \item Type 1: $\exists x^*\in U_{I,j,v}$ such that $f_{x^*, I_1}\geq 2^{l-2}$, or
        \item Type 2: $\exists x^*_1,x^*_2\in U_{I,j,v}$ such that $f_{x^*_1,I_1},f_{x^*_2,I_1}\geq 2^{l-1}\cdot K^{-1/8}$.
    \end{itemize}
    We will analyze the expected number of blocked windows conditioned on each singleton event, then apply the law of total probability.
    
    To this end, we apply Lemma~\ref{lem_KSHH_expectation}, which asserts that the expected number of blocked windows, conditioned on $U_{I,j,v}$, is at most\footnote{Note that Lemma~\ref{lem_KSHH_expectation} is stated in a slightly different form, which is equivalent to the expression here.}
    \[
        \sum_{x\in U_{I,j,v}} \min\left\{f_{x,I}^2\cdot O(2^{-l}\cdot \log n), f_{x,I}\cdot O(K^{-1/8}\cdot \log n)\right\}.
    \]

    Let $\cH_{I_1}$ be the set of heavy elements in $I_1$, $\{x:f_{x,I_1}\geq 2^{l-2}\}$.
    Let $\cH_{I_1,v}\subseteq \cH_{I_1}$ be the set of heavy elements in $I_1$ with hash value $v$, $\{x:h^{(j)}(x)=v,f_{x,I_1}\geq 2^{l-2}\}$.
    Similarly, let $\cM_{I_1}$ be the set of medium heavy elements in $I_1$, $\{x:f_{x,I_1}\geq 2^{l-1}\cdot K^{-1/8}\}$.
    Let $\cM_{I_1,v}\subseteq \cM_{I_1}$ be the set of medium heavy elements in $I_1$ with hash value $v$, $\{x:h^{(j)}=v,f_{x,I_1}\geq 2^{l-1}\cdot K^{-1/8}\}$.
    Note that these sets are fixed given the stream and $h^{(j)}$.
    By Chernoff bound (Theorem~\ref{thm:chernoff}), we have $|\cH_{I_1,v}|\leq O(|\cH_{I_1}|/K+\log n)$ and $|\cM_{I_1,v}|\leq O(|\cM_{I_1}|/K+\log n)$ with probability $1-1/\poly\, n$.
    We assume below that this happens, since it only adds at most $o(1)$ to the expectation otherwise.
    
    For Type 1, fix $x^*\in \cH_{I_1,v}$.
    Since $U_{I,j,v}\supset U_{I_1,j,v}$ and $\Pr\left[x\in U_{I,j,v}\mid h^{(j)}\right]\leq 2^{l}\cdot q$ for any element $x\in U$ such that $h_j(x) =v$, taking expectation over $U_{I,j,v}$, we have 
    \begin{align*}
        &\quad \E_{U_{I,j,v}}\left[\sum_{x\in U_{I,j,v}} \min\left\{f_{x,I}^2\cdot 2^{-l}\cdot \log n, f_{x,I}\cdot K^{-1/8}\cdot \log n\right\}\mid x^*\in U_{I_1,j,v},h^{(j)}\right] \\
        &\leq \E_{U_{I,j,v}}\left[\sum_{x\in U_{I,j,v},x\neq x^*} f_{x,I}^2\,\middle|\, h^{(j)}\right]\cdot 2^{-l}\cdot \log n+f_{x^*,I}\cdot K^{-1/8}\cdot \log n \\
        &\leq 2^l\cdot q\cdot \sum_{x:h^{(j)}(x)=v} f_{x,I}^2\cdot 2^{-l}\cdot \log n+f_{x^*,I}\cdot K^{-1/8}\cdot \log n \\
        &=\sum_{x:h^{(j)}(x)=v} f_{x,I}^2\cdot q\log n+f_{x^*,I}\cdot K^{-1/8}\cdot \log n.
    \end{align*}

    The calculation is similar for Type 2.
    Fix $x^*_1,x^*_2\in\cM_{I_1,v}$, we have that
    \begin{align*}
        &\quad \E_{U_{I,j,v}}\left[\sum_{x\in U_{I,j,v}} \min\left\{f_{x,I}^2\cdot 2^{-l}\cdot \log n, f_{x,I}\cdot K^{-1/8}\cdot \log n\right\}\mid x^*_1,x^*_2\in U_{I_1,j,v}, h^{(j)}\right] \\
        &\leq \sum_{x:h^{(j)}(x)=v} f_{x,I}^2\cdot q\log n+(f_{x^*_1,I}+f_{x^*_2,I})\cdot K^{-1/8}\cdot \log n.
    \end{align*}

    Summing up over the two types and over $v$, we have that
    \begin{align}
        \nonumber&\quad\sum_v\E[\#\textrm{ blocked windows by }A_{I,j,v} \cdot \mathbf{1}(\textrm{$\checkfn(j,v)$ passes on $I_1$})] \\ \nonumber
        &\leq \sum_v\sum_{x^*\in \cH_{I_1,v}}\E[\#\textrm{ blocked windows by }A_{I,j,v} \cdot \mathbf{1}(x^*\in U_{I_1,j,v})] \\ \nonumber
        &\qquad\qquad +\sum_v\sum_{x^*_1,x^*_2\in\cM_{I_1,v}}\E[\#\textrm{ blocked windows by }A_{I,j,v} \cdot \mathbf{1}(x^*_1,x^*_2\in U_{I_1,j,v})] \\ 
        &\leq \sum_v|\cH_{I_1,v}|\cdot 2^{l-1}q\cdot O\left( \sum_{x:h^{(j)}(x)=v} f_{x,I}^2\cdot q\log n\right) \label{eqn:block_I1_1}\\
        &\qquad +\sum_v(|\cM_{I_1,v}|\cdot 2^{l-1}q)^2\cdot O\left( \sum_{x:h^{(j)}(x)=v} f_{x,I}^2\cdot q\log n\right) \label{eqn:block_I1_2}\\ 
        &\qquad +\sum_v\sum_{x^*\in \cH_{I_1,v}} 2^{l-1}q\cdot O(f_{x^*,I}\cdot K^{-1/8}\cdot \log n)\label{eqn:block_I1_3}\\
        &\qquad +\sum_v\sum_{x^*\in\cM_{I_1,v}} |\cM_{I_1,v}|(2^{l-1}q)^2\cdot O(f_{x^*,I}\cdot K^{-1/8}\cdot \log n).\label{eqn:block_I1_4}
    \end{align}
    Note that 
    \[
        |\cH_{I_1}|\leq 2^{-2(l-2)}\cdot \sum_{x} f_{x,I_1}^2\leq O(2^{-2l}\cdot F_{I_1})\leq O(2^{-l}W\log^4 n),
    \]
    and
    \[
        |\cM_{I_1}|\leq K^{1/4}\cdot 2^{-2(l-1)}\cdot \sum_{x} f_{x,I_1}^2\leq O(K^{1/4}\cdot 2^{-2l}\cdot F_{I_1})\leq O(K^{1/4}2^{-l}W\log^4 n).
    \]
    Thus, we have
    \[
        |\cH_{I_1,v}|\leq O\left(|\cH_{I_1}|\cdot K^{-1}+\log n\right)\leq O\left(K^{-1}2^{-l}W\log^4 n+\log n\right),
    \]
    and
    \[
        |\cM_{I_1,v}|\leq O\left(|\cM_{I_1}|\cdot K^{-1}+\log n\right)\leq O\left(K^{-3/4}2^{-l}W\log^4 n+\log n\right).
    \]
    By Cauchy-Schwarz, we have
    \begin{align*}
        \sum_v\sum_{x^*\in \cH_{I_1,v}} f_{x^*,I}&\leq \sqrt{|\cH_{I_1}|\cdot \sum_{x^*\in \cH_{I_1}} f_{x^*,I}^2} \\
        &\leq O\left(\sqrt{2^{-l}W\log^4 n\cdot 2^lW\log^4 n}\right) \\
        &\leq O(W\log^4 n),
    \end{align*}
    and
    \begin{align*}
        \sum_v\sum_{x^*\in \cM_{I_1,v}} f_{x^*,I}&\leq \sqrt{|\cM_{I_1}|\cdot \sum_{x^*\in \cM_{I_1}} f_{x^*,I}^2} \\
        &\leq O\left(\sqrt{K^{1/4}2^{-l}W\log^4 n\cdot 2^lW\log^4 n}\right) \\
        &\leq O\left(K^{1/8}W\log^4 n\right).
    \end{align*}

    Thus, we have
    \begin{align*}
        \eqref{eqn:block_I1_1}&\leq \sum_v O\left(|\cH_{I_1,v}|2^lq^2\log n\cdot F_{I,v}\right) \\
        &\leq \sum_v O\left((K^{-1}2^{-l}W\log^4 n+\log n)2^lq^2\log n\cdot F_{I,v}\right) \\
        &=O\left((K^{-1}q\log^5 n+2^lq^2\log^2 n)F_I\right) \\
        &\leq O\left((K^{-1}q\log^5 n+2^lq^2\log^2 n)2^lW\log^2 n\right) \\
        &\leq O(K^{-1}2^l\log^7 n+2^{2l}q\log^4 n) \\
        &\leq O(2^l/\log^4 n), \\
        \eqref{eqn:block_I1_2}&\leq \sum_v O\left(\left((K^{-3/4}2^{-l}W\log^4 n+\log n)2^lq\right)^2\cdot q\log n\cdot F_{I, v}\right) \\
        &\leq O\left((K^{-3/2}q\log^9 n+2^{2l}q^3\log^3 n)\cdot F_I\right) \\
        &\leq O\left((K^{-3/2}q\log^9 n+2^{2l}q^3\log^3 n)\cdot 2^lW\log^4 n\right) \\
        &\leq O\left(K^{-3/2}2^l\log^{13} n+2^{3l}q^2\log^7 n\right) \\
        &\leq O(2^l/\log^4 n),\\
        \eqref{eqn:block_I1_3}&\leq O\left(K^{-1/8}2^lqW\log^5 n\right) \\
        &\leq O\left(2^l/\log^4 n\right), \\
        \eqref{eqn:block_I1_4}&\leq \sum_v\sum_{x^*\in\cM_{I_1,v}}O\left((K^{-3/4}2^{-l}W\log^4 n+\log n)2^{2l}q^2\cdot f_{x^*,I}\cdot K^{-1/8}\log n\right) \\
        &\leq O\left((K^{-3/4}2^{-l}W\log^4 n+\log n)2^{2l}q^2\cdot K^{1/8}W\log^4 n\cdot K^{-1/8}\log n\right) \\
        &\leq O\left(K^{-3/4}2^{l}\log^9 n+2^{2l}q\log^6 n\right) \\
        &\leq O(2^l/\log^4 n).
    \end{align*}
    The four bounds combined imply that
    \[
        \sum_v\E[\#\textrm{ blocked windows by }A_{I,j,v} \cdot \mathbf{1}(\textrm{$\checkfn(j,v)$ passes on $I_1$})] \leq O(2^l/\log^4 n).
    \]
    A similar bound can be obtained for $I_2$.
    Since running $A_{I,j,v}$ requires $\checkfn(j,v)$ to pass either on $I_1$ or $I_2$, we prove the lemma.
\end{proof}

\begin{lemma}\label{lem:block_window_concentrate}
    In the hypothetical version, for each $j\in[J]$, all $K$-super heavy hitter algorithms initiated by hash function $j$ block at most $O(1/\log^4 n)$-fraction of the windows in total, with probability $1-1/\poly\, n$.
\end{lemma} 
\begin{proof}
We first consider $l$ such that $2^lW<n/\log^{10} n$, and bound the total number of blocked windows by all $A_{I,j,v}$ for $I$ of size $2^l$.
    Fix such a node $I$ with children nodes $I_1,I_2$.
    If $F_{I_1},F_{I_2},F_I\leq O(2^lW\log^4 n)$, then by Lemma~\ref{lem:block_exp}, the expected total number of blocked windows by $A_{I,j,v}$ for all $v$ is $O(2^l/\log^4 n)$ \emph{conditioned on} $h^{(j)}$.
    Denote this random variable by $X_{I,j}$.
    Note that conditioned on $h^{(j)}$, $X_{I,j}$ only depends on the random sets $S_i^{(j)}$ (as well as the randomness used by $A_{I,j,v}$), and for different $I$ of size $2^l$, they depend on different sets $S_i^{(j)}$. 
    Hence, $X_{I,j}$ are independent for all $I$ of the same size $2^l$ \emph{conditioned on} $h^{(j)}$.

    Since there can be at most $O(n/2^lW\log^4 n)$ $I$'s with at least one of $F_{I_1},F_{I_2},F_I$ more than $\Omega(2^lW\log^4 n)$, and $X_{I,j}\leq |I|$, we have
    \[
        \E\left[\sum_{I:|I|=2^l} X_{I,j}\mid h^{(j)}\right]\leq O(n/2^lW\log^4 n)\cdot 2^l+n/2^lW\cdot O(2^l/\log^4 n)=O(n/(W\log^4 n)).
    \]
    By Theorem~\ref{thm:hoeffding's}, and the fact that $X_{I,j}\leq 2^l$,
    we have $\sum_{I:|I|=2^l} X_{I,j}\leq O(n/(W\log^4 n))$ except with probability $1/\poly\, n$. 

    Now for $l$ such that $2^lW\geq n/\log^{10}n$, we will apply Lemma~\ref{lem:low_freq_conc} and Lemma~\ref{lem:high_freq_conc} to bound the number of blocked windows.
    Fix such a node $I$ of length $2^l$.
Let $I_1$ and $I_2$ be its children. We claim the following.
\begin{claim}\label{clm:blk_window_I}
    For each $(j,v)$, the number of windows blocked by $A_{I,j,v}$ is at most $2^l/K^{1/17}$ with probability at least $1-1/\poly\, n$. 
\end{claim}
Before we prove this claim, let us see why this claim suffices to prove the desired bound. Consider the definitions of $H_{I_1},H_{I_1,v}, M_{I_1}, M_{I_1,v}$ from the proof of \cref{lem:block_exp}. We execute $A_{I,j,v}$ only if $\checkfn(j,v)$ passes round $l-1$ on $I_1$ or on $I_2$. Recall that $\checkfn(j,v)$ passes round $l-1$ on $I_1$ only if either
    \begin{itemize}
        \item Type 1: $\exists x^*: h^{(j)}(x^*) = v$ and $f_{x^*, I_1}\geq 2^{l-2}$, or
        \item Type 2: $\exists x^*_1,x^*_2: h^{(j)}(x^*_1) = h^{(j)}(x^*_2) = v$ and $f_{x^*_1,I_1},f_{x^*_2,I_1}\geq 2^{l-1}\cdot K^{-1/8}$.
    \end{itemize}

We now bound the total number of check functions that are running on $I$ for $\checkfn(j,v)$ that passes round $l-1$ on $I_1$. Recall that for fixed $j$ and $v$, we run only a single instance of $\checkfn(j,v)$ on $I$. Therefore, the total number of check functions (for each $j$) is at most
\[ \sum_v \mathbbm{1}[\exists x^*: h^{(j)}(x^*) = v \text{ and } f_{x^*, I_1}\geq 2^{l-2}] + \mathbbm{1}[\exists x_1^*, x_2^*: h^{(j)}(x_1^*) = h^{(j)}(x_2^*) = v \text{ and } f_{x_1^*, I_1}, f_{x_2^*, I_1}\geq 2^{l-1} \cdot K^{-1/8}] \]
Observe that $ \sum_v \mathbbm{1}[\exists x^*: h^{(j)}(x^*) = v \text{ and } f_{x^*, I_1}\geq 2^{l-2}] = |H_{I_1}|\le O(\log^{20} n)$. As for the second term, we want to bound the number of hash values $v$ such that at $|M_{I_1,v}|\ge 2$. Since $2^lW<n/\log^{10} n$, we have $|M_{I_1}|\le O(K^{1/4}\log^{20} n)$.
For any $t\ge 1$, the probability that there are $t$ hash values with $|M_{I_1,v}|\ge 2$ is at most
\[\binom{|M_{I_1}|}{2t} \cdot 2^t\cdot t! \cdot \left(\tfrac{1}{K}\right)^t \le K^{-t/3} \, .\] Since this probability is exponentially small in $t$, the number of hash values $v$ such that at $|M_{I_1,v}|\ge 2$ is at most $O(\log n)$ with probability at least $1-1/\poly\, n$. The analysis is similar for $I_2$.

Combining this with Claim~\ref{clm:blk_window_I} and summing and taking a union bound over all the nodes with $2^lW\geq n/\log^{10}n$, we conclude that at most $O(1/\log n)$-fraction of the windows are blocked in total, with probability at least $1-1/\poly\, n$. We will now prove the claim by applying \cref{lem:high_freq_conc} and \cref{lem:low_freq_conc}. Consider an instance of $\checkfn(j,v)$ that is initiated on $I$. In order to apply \cref{lem:high_freq_conc}, we need to bound the total number of occurrences of elements with hash value $v$ and frequency at least $2^l/K^{1/8}$ in $I$. The total number of such elements is at most $O(K^{1/4} \log^{20} n)$. The hash function $h^{(j)}$ assigns each of these elements to hash value $v$ independently at random with probability $1/K$. Applying the additive form of Chernoff bound, we conclude that the total number of elements assigned to hash value $v$ is at most $O(\log n)$ with probability at least $1-1/\poly\, n$. 
Since each element has frequency at most $O(\sqrt{n})$, we conclude that the total number of occurrences of elements with hash value $v$ and frequency at least $2^l/K^{1/8}$ in $I$, is at most $O(\sqrt{n} \log n) \le O\left(2^l \left(\tfrac{W}{\sqrt{n}}\right) \log^{11} n\right)$, with probability at least $1-1/\poly\, n$. By applying \cref{lem:high_freq_conc}, we conclude that the number of windows blocked due to elements of frequency at least $2^l/K^{1/8}$ is at most $2^l/K^{1/17}$ with probability at least $1-1/\poly\, n$.

To apply \cref{lem:low_freq_conc}, we need to bound the second moment of the frequencies of elements which occur less than $2^l/K^{1/8}$ times in $I$. For an element $a$ of frequency at most $2^l/K^{1/8}$, let $X_a$ be the indicator random variable for the event that $h^{(j)}(a) = v$ and let $Y_a = f_{a,I}^2 X_a$. We want to bound $Y=\sum_{a:f_{a,I}< 2^l/K^{1/8}} Y_a$. Note that 
\[E[Y] = (1/K)\sum_{a:f_{a,I}< 2^l/K^{1/8}} f_{a,I}^2 \le n/K\, .\]
Since $0\le Y_a \le f_a^2$, we can apply Hoeffding's inequality (\cref{thm:hoeffding's}) to conclude that with probability at least $1-1/\poly(n)$, $Y\le n/K^{1/8}$. Now, by applying \cref{lem:low_freq_conc}, we conclude that the number of windows blocked due to elements of frequency less than $2^l/K^{1/8}$ is at most $2^l/K^{1/17}$ with probability at least $1-1/\poly\, n$. Thus, we have proved our claim.
\end{proof}


The main lemma that leads to the final claim on the success probability of the algorithm is the following.
We focus on one $j$, and fix the random bits used by all $j'\neq j$, i.e., hash function $h^{(j')}$, the sets $S_i^{(j')}$, and the random bits used by $K$-SHH instances that are initiated by some instance of $\check(j', v)$.
Then we will analyze regardless of fixing of these random bits, we can always find the correct hash value $v_j=h_j(x^*)$ with some constant probability (only over the randomness used by $j$).

\begin{lemma}\label{lem:prob_hash_value}
    For each $j$, conditioned on all the random bits used by $j'\neq j$, the probability that we find the correct $v_j=h_j(x^*)$ is at least $1/25$, and we find an incorrect hash value with probability $O(K^{-1/2})$. 
\end{lemma}
\begin{proof}
    Fix the random bits used by all $j'\neq j$.
    This fixes the set of windows blocked by $j'\neq j$ in the hypothetical algorithm.
    By Lemma~\ref{lem:block_window_concentrate}, at most $1/\log^4 n$-fraction of the windows are blocked in total.
    Recall that the windows blocked by the real algorithm are a subset, the remaining at least $(1-1/\log^4 n)$-fraction of the windows are guaranteed to not be blocked by $j'\neq j$, regardless of the randomness used by hash function $j$.

    For $j$ fixed, the algorithm records the first $(j, v)$ such that $\checkfn(j, v)$ returns YES.
    We will prove 
    \begin{itemize}
        \item for all $v\neq h^{(j)}(x^*)$, where $x^*$ is the heavy hitter, $\checkfn(j, v)$ will return YES with probability $O(K^{-1/2})$;
        \item for $v=h^{(j)}(x^*)$, $\checkfn(j, v)$ will be initiated and will return YES with probability $\Omega(1)$.
    \end{itemize} 

    For the first bullet point, consider the last round of $\checkfn$, i.e., $l=\maxl$.
    Consider a node $I$ in level $l$ of the execution forest, the length of $I$ is $2^l=C\sqrt{n}/8$.
    Fix $v\neq h^{(j)}(x^*)$, we now upper bound the probability that $\checkfn(j, v)$ passes the last round on $I$.
    First observe that no element other than the heavy-hitter in $U_{I,j,v}$ can have frequency at least $2^{l-1}=C\sqrt{n}/16$.
    Hence, by Lemma~\ref{lem:KSHH_no}, this requires at least two elements in $U_{I,j,v}$ with frequency at least $2^l\cdot K^{-1/8}=C\sqrt{n}K^{-1/8}/8$.
    But there can be at most $O(K^{1/4})$ elements with this frequency, even in the whole stream.
    The probability there exists two of them with hash value $v$ is at most $O\left(\left(\frac{K^{1/4}}{K}\right)^2\right)=O\left(K^{-3/2}\right)$.
    In particular, $\checkfn(j, v)$ passes the last round on $I$ with probability at most $K^{-3/2}$.
    There are $O(1)$ nodes in the last level, taking a union bound over $v$ and the nodes in this level, we conclude that the algorithm records a wrong hash value for $j$ with probability at most $O(K^{-1/2})$.

    In the remainder of the proof, we will show that we will sample and start $\checkfn(j, v)$ from a ``good'' window with constant probability, and a $\checkfn(j, v)$ that is started from a ``good'' window will return YES with constant probability.

    Consider a window $i$, suppose we start $\checkfn$ in parallel at the end of this window.
    Let $I_0,I_1,\ldots,I_{\maxl}$ be the nodes in the execution forest that this check will run $K$-SHH on, where $I_l$ is in level $l$ (in particular, $I_0$ consists of window $i+1$).
    We say window $i$ is ``good'' if
    \begin{enumerate}
        \item $I_0,\ldots,I_{\maxl}$ exist, i.e., $i$ is not close to the end of the stream,
        \item for each $l=0,\ldots,\maxl$, $I_l$ has at most $2^l\log^{-1} n$ blocked windows, and
        \item for each $l=0,\ldots,\maxl$, the second moment of $I_l$ is at most $2^lW\log^3 n$, i.e., $\sum_x f_{x, I_l}^2\leq 2^lW\log^3 n$.
    \end{enumerate}
    Note that the definition of a good window depends on the randomness used when checking other hash functions $j'\neq j$ (for the second condition), but not the randomness used by $j$.

     For each good window, $v=h^{(j)}(x^*)$, and $l=0,\ldots,\maxl$, we have
    \[
        \mathbb{E}\left[\sum_{x\in U_{I_l,j,v},x\neq x^*} f_{x,I_l}^2\right] \leq \Pr[x\in U_{I_l,j,v}] \cdot  2^lW\log^3 n \le (2^l q) \cdot 1/K \cdot 2^lW\log^3 n = 4^l\log^3 n/K.
    \] Therefore, by Markov's inequality and taking the union bound over all $l$, we have that with probability at least $1-o(1)$, for every  $l=0,\ldots,\maxl$, 
    \[ \sum_{x\in U_{I_l,j,v},x\neq x^*} f_{x,I_l}^2 \leq 4^l\log^5 n/K\, .\]
    
    Assuming the bound holds, if we start $\checkfn(j, v)$ for $v=h^{(j)}(x^*)$ from a good window, then for each round $l$, $K$-SHH returns YES with probability $1-1/\poly\, n$ by Lemma~\ref{lem:Yes_case}.
    Taking a union bound over all levels, $\checkfn$ returns YES except with $o(1)$ probability.
    It remains to show that we will start such a $\checkfn$ with constant probability.
    
    We show that at least a constant fraction of the windows are good.
    \begin{claim}\label{clm:good_window}
         With probability $1-1/\poly\, n$, at least $1/2-o(1)$ fraction of the windows are good.
    \end{claim}
    To see this, we analyze the fraction of windows that is \emph{not} good due to the violation of each of the conditions above.
    \begin{enumerate}
        \item Note that $2^{\maxl}=C\sqrt{n}/8=n/8W$.
        As long as $i+4\cdot 2^{\maxl}\leq n/W$, i.e., $i\leq n/2W$, $I_0,\ldots,I_{\maxl}$ exist.
        At most half of the windows violate this condition.
        \item By applying \cref{lem:block_window_concentrate} and taking the sum over all hash functions $j'\ne j$, we observe that at most $O\left((n/W)\cdot \log^{-3} n\right)$ windows are blocked with probability $1-1/\poly\, n$. For a fixed $l$, at most $\frac{O((n/W)\cdot \log^{-3} n)}{2^l\log^{-1}n}=O\left((n/2^lW)\log^{-2} n\right)$ nodes in level $l$ can violate the condition.
        For each such node, there are exactly $2^l$ windows $i$ that could lead to this node in level $l$.
        At most $O((n/W)\log^{-2} n)$ windows violate this condition for each $l$, thus, a total of $O((n/W)\log^{-1} n)$ windows violate this condition.
        This is $o(1)$-fraction.
        \item Similarly, since the second moment of the stream is $O(n)$, fix $l$, at most $\frac{O(n)}{2^lW\log^3 n}=O((n/2^lW)\log^{-3} n)$ nodes in level $l$ can violate this condition.
        At most $O((n/W)\log^{-3} n)$ windows violate this condition for each $l$.
        Thus, a total of $o(1)$-fraction windows violate this condition.
    \end{enumerate}

    Given the claim, we now prove that with $\Omega(1)$ probability, the (actual) algorithm will start $\checkfn(j,v)$ for the correct $v$ from some ``good'' window.
    Concretely, we prove the following claim.
    \begin{claim}\label{clm:good_start}
        With $1/20$ probability, there exists a \emph{good} window $i$ such that
        \begin{itemize}
            \item $x^*$ is in $S_i^{(j)}$, and it is the only element in $S_i^{(j)}$ that appears in window $i$;
            \item at the end of window $i$, at most $100$ $\checkfn(j,\cdot)$ are still running.
        \end{itemize}
    \end{claim}
    It is easy to verify that when the event in the claim happens, we start $\checkfn(j, v)$ for the correct $v$ from that window, and hence, the lemma holds.
    
    First observe that for each window $i$, the probability that the first condition holds is at least
    \[
        q\cdot (1-q)^W\approx q(1-1/C^2).
    \]
    To prove the claim, we need to prove that this happens in a good window and when there are at most $100$ other $\checkfn$ running.
    We apply Lemma~\ref{lem:bound_checkfn_time}, which asserts that at most $1/4+o(1)$ of the windows in expectation can have more than $100$ $\checkfn(j,\cdot)$ running.
    Combined with Claim~\ref{clm:good_window}, in expectation, at least $1/4-o(1)$-fraction of the windows are simultaneously good and have at most $100$ $\checkfn$ running.
    Ideally, we would like to argue that in at least one of these $(1/4-o(1))n/W$ windows, we sample $x^*$ in $S_i^{(j)}$ but not other elements with constant probability by bounding it by $1-(1-q/3)^{(1/4-o(1))n/W}=\Omega(1)$.
    However, we cannot bound it in this way, since the set of this $1/4-o(1)$-fraction of windows can depend on the randomness used by hash function $j$ (in particular, the sets $S_i^{(j)}$).
    To this end, we will prove the following claim by induction, which implies the bound we aim for.


     Let $G_i$ be the random variable indicating the number of windows that are simultaneously good and have at most $100$ $\checkfn$ running \emph{from window $i+1$ to the end}.
     We have $\E[G_0]\geq (1/4-o(1))n/W$.
     Let $R_i$ be the randomness used by the algorithm in window $i$, i.e., the random sets $S^{(j)}_1,\ldots,S^{(j)}_i$, and all the random bits used by $K$-SHH in window $i$.
   
    \begin{claim}
     Fix all randomness used by hash functions $j'\neq j$.
     Let $E$ be the (bad) event that we \emph{don't} have a window $i$ that is simultaneously good and have at most $100$ $\checkfn$ running such that $x^*$ is in $S_{i}^{(j)}$, and it is the only element in $S_{i}^{(j)}$ that appears in the window.
     Then we have
     \[
        \Pr[E\mid h^{(j)}, R_{\leq i}] \leq \gamma\left(\frac{q}{2}\cdot \E[G_i\mid h^{(j)}, R_{\leq i}]\right),
     \]
     where $\gamma(y):=\exp(-y/3-5y^2/9)$.
    \end{claim}
    As we will use it later, observe that the function $\gamma$ on the RHS is concave.
    This is because $\frac{1}{2}q\cdot \E[G_i\mid h^{(j)}, R_{\leq i}]\leq \frac{1}{2}q\cdot n/W\leq 1/2$, and for $y\in[0,1/2]$, the second derivative of $\gamma(y)=e^{-y/3-5y^2/9}$ is $((1/3+10y/9)^2-10/9)e^{-y/3-5y^2/9}<0$.

    Now to see this claim is sufficient for proving Claim~\ref{clm:good_start}, we set $i=0$.
    We have that 
    \[
        \Pr[E\mid h^{(j)}]\leq \gamma\left(\frac{q}{2}\cdot \E[G_0\mid h^{(j)}]\right).
    \]
    Taking expectation over $h^{(j)}$ and using the concavity and monotonicity of $\gamma$, we have
    \begin{align*}
        \Pr[E]&\leq \E_{h^{(j)}}\left[\gamma\left(\frac{q}{2}\cdot \E[G_0\mid h^{(j)}]\right)\right] \\
        &\leq \gamma\left(\frac{q}{2}\cdot \E[G_0]\right) \\
        &\leq \gamma\left(\frac{q}{2}\cdot (1/4-o(1))n/W\right) \\
        &\leq \gamma\left(1/8-o(1)\right) \\
        &\leq \exp(-1/24-5/576+o(1)) \\
        &\leq 19/20.
    \end{align*}
    That is, the bad event $E$ happens with probability at most $19/20$, we must start $\checkfn$ on the correct hash value $v=h^{(j)}(x^*)$ with probability at least $1/20$.
    This implies Claim~\ref{clm:good_start}.

        
    It remains to prove the claim.
    We prove by induction on $i$ from $i=n/W$ to $0$.
    The base case is when $i=n/W$.
    The bound trivially holds, since $G_i=0$, the RHS is equal to 1. 
    Now suppose the bound holds for $i+1$, and let us consider $i$.
    There are two cases: whether window $i+1$ is simultaneously good and has at most $100$ $\checkfn$ running.
    Note that this is completely determined by $h^{(j)}$ and $R_{\leq i}$.
    \begin{itemize}
        \item Suppose $h^{(j)}$ and $R_{\leq i}$ are such that window $i+1$ is \emph{either not} good \emph{or} has \emph{more than} $100$ $\checkfn$ running.
        Then we have $G_{i+1}=G_i$ conditioned on $h^{(j)}$ and $R_{\leq i}$.
        In this case, we have
        \begin{align*}
            \Pr[E\mid h^{(j)}, R_{\leq i}]&\leq \E_{R_{i+1}\mid h^{(j)}, R_{\leq i}}\left[\Pr[E\mid h^{(j)}, R_{\leq i+1}]\right] \\
            &\leq \E_{R_{i+1}\mid h^{(j)}, R_{\leq i}}\left[\gamma\left(\frac{q}{2}\cdot \E[G_{i+1}\mid h^{(j)}, R_{\leq i+1}]\right)\right] \\
            &= \E_{R_{i+1}\mid h^{(j)}, R_{\leq i}}\left[\gamma\left(\frac{q}{2}\cdot \E[G_{i}\mid h^{(j)}, R_{\leq i+1}]\right)\right] \\
            \intertext{which by the concavity, is}
            &\leq \gamma\left(\frac{q}{2}\cdot \E_{R_{i+1}\mid h^{(j)}, R_{\leq i}}\left[\E[G_{i}\mid h^{(j)}, R_{\leq i+1}]\right]\right) \\
            &=\gamma \left(\frac{q}{2}\cdot \E[G_{i}\mid h^{(j)}, R_{\leq i}]\right),
        \end{align*}
        proving the bound for $i$.
        \item Suppose $h^{(j)}$ and $R_{\leq i}$ are such that window $i+1$ is good and has at most $100$ $\checkfn$ running.
        Then $G_{i+1}=G_i-1$ conditioned on $h^{(j)}$ and $R_{\leq i}$.
        Note that $R_{i+1}$ is such that $x^*\in S_{i+1}^{(j)}$ and it is the only element in window $i+1$ that is sampled in $S_{i+1}^{(j)}$, then event $E$ does not happen.
        The set $S_{i+1}^{(j)}$ is independent of $h_i^{(j)}$ and $R_{\leq i}$, and hence, it happens with probability at least $q(1-1/C^2)$ as we argued above.
        Denote by $F_{i+1}$ the event that $x^*\in S_{i+1}^{(j)}$ and it is the only element in window $i+1$ that is sampled in $S_{i+1}^{(j)}$.
        Thus, we have
        \begin{align*}
            \Pr[E\mid h^{(j)}, R_{\leq i}]&\leq \E_{R_{i+1}\mid h^{(j)},R_{\leq i}}\left[\Pr[E\mid h^{(j)}, R_{\leq i+1}]\right] \\
            &\leq \E_{R_{i+1}\mid h^{(j)},R_{\leq i}}\left[\mathbf{1}_{\neg {F}_{i+1}}\cdot \gamma\left(\frac{q}{2}\cdot \E[G_{i+1}\mid h^{(j)}, R_{\leq i+1}]\right)\right] \\
            &=\Pr[\neg F_{i+1}\mid h^{(j)},R_{\leq i}] \cdot \E_{R_{i+1}\mid h^{(j)},R_{\leq i},\neg F_{i+1}}\left[\gamma\left(\frac{q}{2}\cdot \E[G_{i+1}\mid h^{(j)}, R_{\leq i+1}]\right)\right] \\
            &\leq (1-q(1-1/C^2))\cdot \gamma\left(\frac{q}{2}\cdot \E_{R_{i+1}\mid h^{(j)},R_{\leq i},\neg F_{i+1}}\left[\E[G_{i+1}\mid h^{(j)}, R_{\leq i+1}]\right]\right) \\
            &=(1-q(1-1/C^2))\cdot \gamma\left(\frac{q}{2}\cdot \E[G_{i+1}\mid h^{(j)}, R_{\leq i},\neg F_{i+1}]\right) \\
            &\leq e^{-q(1-1/C^2)}\cdot \gamma\left(\frac{q}{2}\cdot \left(\E[G_{i}\mid h^{(j)}, R_{\leq i},\neg F_{i+1}]-1\right)\right).
        \end{align*}
        Again by the fact that $\Pr[F_{i+1}\mid h^{(j)},R_{\leq i}]\approx q(1-1/C^2)$ and $G_i\in [0, n/W]$, we have that
        \[
            \E[G_i\mid h^{(j)}, R_{\leq i}, \neg F_{i+1}]\geq \E[G_i\mid h^{(j)}, R_{\leq i}]-(q(1-1/C^2))\cdot n/W=\E[G_i\mid h^{(j)}, R_{\leq i}]-(1-1/C^2).
        \]
        Hence, by the monotonicity of $\gamma$,
        \begin{align*}
             \Pr[E\mid h^{(j)}, R_{\leq i}]&\leq e^{-q(1-1/C^2)}\cdot\gamma\left(\frac{q}{2}\cdot\left(\E[G_i\mid h^{(j)}, R_{\leq i}]-(1-1/C^2)-1\right)\right) \\
             &\leq e^{-q(1-1/C^2)}\cdot\gamma\left(\frac{q}{2}\cdot\E[G_i\mid h^{(j)}, R_{\leq i}]-q(1-1/2C^2)\right).
        \end{align*}
        On the other hand,
        \begin{align*}
            &\gamma\left(\frac{q}{2}\cdot\left(\E[G_i\mid h^{(j)}, R_{\leq i}]\right)-q(1-1/2C^2)\right) \\
            &=\exp\left(-\frac{q}{6}\cdot\E[G_i\mid h^{(j)}, R_{\leq i}]+\frac{q}{3}(1-1/2C^2)-\frac{5}{9}\left(\frac{q}{2}\cdot\E[G_i\mid h^{(j)}, R_{\leq i}]-q(1-1/2C^2)\right)^2\right) \\
            &\le \gamma\left(\frac{q}{2}\cdot\left(\E[G_i\mid h^{(j)}, R_{\leq i}]\right)\right)\cdot \exp\left(\frac{q}{3}(1-1/2C^2)+\frac{5}{9}\left(q\cdot\E[G_i\mid h^{(j)}, R_{\leq i}]\cdot q(1-1/2C^2)\right) \right) \\
            &\le\gamma\left(\frac{q}{2}\cdot\left(\E[G_i\mid h^{(j)}, R_{\leq i}]\right)\right)\cdot \exp\left(\frac{q}{3}(1-1/2C^2)+\frac{5q}{9}(1-1/2C^2)\cdot(qn/W)\right)\\
            &\le \gamma\left(\frac{q}{2}\cdot\left(\E[G_i\mid h^{(j)}, R_{\leq i}]\right)\right)\cdot \exp\left(q(1-1/C^2)\right)\, .
        \end{align*}
        This completes the induction step for the second case, finishing the proof.
    \end{itemize}

\end{proof}

\begin{theorem}
    The space usage of Algorithm~\ref{alg:ell2} is $O(\log n)$ bits.
    It finds $x^*$ with probability $1-2^{-\Omega(\log n/\log\log n)}$.
\end{theorem}
\begin{proof}
    It is straightforward to verify that the space usage is $O(\log n)$ bits: at most $O(J)=O(\log n/\log K)$ $\checkfn$ and $K$-SHH may be running in parallel, each uses $O(\log K)$ bits regularly; at most $O(1)$ instances of $K$-SHH may use $O(\log n)$ bits at the same time.

    Now we lower bound the probability that it finds $x^*$.
    Let $X_j$ indicate if we find $v_j=h^{(j)}(x^*)$.
    Lemma~\ref{lem:prob_hash_value} implies that $\Exp[X_j\mid X_1,\ldots,X_{j-1}]\geq 1/25$.
    The expectation of $\sum X_j$ is at least $J/25=\frac{128}{25}\log U/\log K$.
    Thus, by \cref{thm:azuma-hoeffing}, $\Pr[\sum_j X_j\geq 4\log U/\log K]\geq 1-2^{-\Omega(\log U/\log K)}=1-2^{-\Omega(\log n/\log\log n)}$.

    When this happens, $x^*$ satisfies the last line of Algorithm~\ref{alg:ell2}.
   On the other hand, it is the only such element with high probability. Consider any $x\neq x^*$ that matches $4\log U/\log K$ hash values recorded. For each $v_j$ it matches, either $h^{(j)}(x)=h^{(j)}(x^*)$ or $v_j\neq h^{(j)}(x^*)$. 
    The former happens with probability $K^{-1}$, and in the latter case, it is the hash value found by $j$ with probability $O(K^{-1/2})$ by Lemma~\ref{lem:prob_hash_value}.
    By union bounding over the choice of $4\log U/\log K$ indices $j$, the probability that this happens is at most
    \[
        \binom{J}{4\log U/\log K}\cdot O(K^{-1/2})^{4\log U/\log K}\leq O(K^{-1/2})^{4\log U/\log K}\leq 1/U^{2-o(1)}.
    \]
    Thus, the probability that there exists a different $x\neq x^*$ that also satisfies the last line of the algorithm is at most $1/U^{1-o(1)}$.
    This proves the theorem. 
    
\end{proof}
\section{Streaming algorithm for detecting $\ell_2$ $\epsilon$-heavy hitters}\label{sec:epsilon-heavy}

In this section, we present the algorithm for finding $\ell_2$ $\epsilon$-heavy hitters with high probability using $O\left(\frac{\log n}{\epsilon}\right)$-bits of space, assuming each $\epsilon$-\emph{heavy hitter} is evenly distributed in the stream.
Note that the space is tight because there could be $\theta\left(\frac{1}{\epsilon}\right)$ $\epsilon$-heavy hitters and we would need at least $O\left(\frac{\log n}{\epsilon}\right)$-bits of space to output them.
Here again, we assume that the length of the stream is $n$, the second moment of the frequency vector is $O(C^2 n)$, and we assume access to free randomness.

The main changes to \cref{alg:ell2} from the previous section are highlighted in {\color{red}red} below.

\begin{algorithm}[H]\caption{Algorithm for $\ell_2$ $\epsilon$-heavy hitters}\label{alg:ell2_eps}
\begin{algorithmic}[1]
\Input a stream of length $n$ of elements from set $U$ of size $\poly\, n$, and a parameter $0<\epsilon\le 1$ with $\epsilon =\omega\left(1/n\right)$.
\State \emph{assumptions: for every $\epsilon$-heavy hitter, if we divide the stream into windows of length {\color{red}$W:=\sqrt{n/\epsilon}/C$}, there is at least one occurrence of it in every window; $\ell_2^2$ of the frequency vector is at most {\color{red}$C^2 n$}.}
\Output the set of all $\epsilon$-heavy hitters in the stream
\State fix parameters {\color{red}$q=\frac{\sqrt{\epsilon}}{C\sqrt{n}}$}, {\color{red}$K=\left(\log n/\epsilon\right)^{1000}$}, and {\color{red}{$J=128\log U/\epsilon \log K$}}
\State let $h^{(1)},\ldots,h^{(J)}$ be independent random hash functions $h^{(j)}: U\rightarrow [K]$
\State for each window $i$ and each $j\in[J]$, sample a random set $S^{(j)}_i\subseteq U$ such that each $x\in U$ is in $S^{(j)}_i$ with probability $q$ independently \emph{using free random bits}
\For{each window $i$}
\State for $j\in[J]$, if there is some $x\in S^{(j)}_i$ that appears in the window, record $(j, h^{(j)}(x))$ (if there is more than one such $x$ for some $j$, record the first one)
\For {each hash $(j, v)$ recorded at the end of the window}
\If {at most $100$ instances of $\checkfn$ are still running for this $j$}
\State start $\checkfn(j, v)$ \emph{in parallel}
\EndIf
\EndFor
\EndFor
\State For each $j$, record the first $(j, v)$ such that $\checkfn(j, v)$ returned YES
\State If $x$ such that $h^{(j)}(x)=v$ for at least $4\log U/\log K$ of them, return $x$ (If there are multiple such $x$, return all of them)
\end{algorithmic}
\end{algorithm}

The \emph{execution forest} of \cref{alg:ell2_eps} would be similar to the execution forest that we defined in the previous section for \cref{alg:ell2} with the new $\maxl=\log(C\sqrt{\epsilon n}/8)$.

Below we describe the algorithm for the $\checkfn$ function and highlight the changes from \cref{alg:check} in {\color{red}red}. Note the the only changes are in the value of $\maxl$ and the maximum number of K-SHH algorithms that can use $\log n$ bits simultaneously.

\begin{algorithm}[H]\caption{$\checkfn(j,v)$ checks if $h^{(j)}(x^*)=v$  (when $x^*$ exists)}\label{alg:check_eps}
\begin{algorithmic}[1]
    \State let {\color{red}$\maxl=\log(C\sqrt{\epsilon n}/8)$}
    \For {$l=0,\ldots,\maxl$}
	\State currently at the beginning of a level-$l$ node $I$ in the execution forest, let it be the beginning of the $i$-th window such that $2^l\mid i$; the current $\checkfn$ must be initiated from a window in the range $[\max\{0, i-2^{l+1}+1\}, i-2^l]$
        \State\label{line:unique_jv} if there is another instance of $\checkfn$ with the same $j$ and $v$ is about to start the same round $l$ (hence on the same $I$), then keep only one instance and terminate the others
        \State let $U_{I, j}\leftarrow \bigcup_{k\in [\max\{0, i-2^{l+1}+1\}, i-2^l]} S_k^{(j)}$ be the union of all $S_k^{(j)}$ of windows that may have initiated this $\checkfn$
	\State run $K$-SHH (\cref{Algorithm 1}) on $I$, restricted to elements in $U_{I, j}$ with hash value $v$ (i.e., set the parameter $V$ to $U_{I,j,v}$, $n$ is the length of the original stream, and we assume that the universe size is at most $\poly\,  n$, $W$ and $K$ are fixed parameters that do not vary across the different executions of $K$-SHH)
 $$U_{I, j, v}:=\{x: h^{(j)}(x)= v\}\cap U_{I, j},$$ 
 denote this instance of $K$-SHH by $A_{I,j,v}$; {\color{red}\textbf{at every time-step, we ensure that there are at most $1/\epsilon$ different $j$'s such that their $K$-SHH instances use $O(\log n)$ bits (and they \emph{block} $A_{I',j',v'}$ for every other $j'$, i.e., those instances do not proceed after \cref{algo_KSHH_line-block} in \cref{Algorithm 1})}\label{alg_check_line-6}}
	\If {$K$-SHH returns NO}
            \State return NO
        \EndIf
        \State wait until the start of the next window $i$ such that $2^{l+1}\mid i$
    \EndFor
    \State return YES
\end{algorithmic}
\end{algorithm}

Like in the \cref{sec:SHH}, in order to analyze the above algorithm, we would need to work with \emph{hypothetical} versions of it.
Observe that the main changes to the algorithms in the previous section are the setting of the parameters $q,W,K,J$, and $\maxl$. In addition, $1/\epsilon$ many K-SHH algorithms are now allowed to use $\log n$ bits in parallel, and \cref{alg:ell2_eps} outputs any element whose hash value passes $\checkfn$ in at least $0.5 \epsilon$ fraction of the hash functions.
Given this, it can be verified that \cref{lem:bound_checkfn_time,lem:block_exp,lem:block_window_concentrate}, and the claims within the lemmas still hold with the new setting of the parameters for the $\epsilon$-heavy hitters.

We now prove a lemma similar to \cref{lem:prob_hash_value} in the previous section.

\begin{lemma}\label{lem:prob_hash_value_eps}
 For each $j$, for every $\epsilon$-heavy hitter $x^*$, conditioned on all the random bits used by $j'\neq j$, the probability that we find $v_j=h_j(x^*)$ is at least $\epsilon/25$. On the other hand, we find a hash value $v$ such that $v\ne h_j(x)$ for any $\epsilon/256$-heavy hitter with probability $O(K^{-1/3})$.
\end{lemma}

   Fix the random bits used by all $j'\neq j$.
    This fixes the set of windows blocked by $j'\neq j$ in the hypothetical algorithm.
    By Lemma~\ref{lem:block_window_concentrate}, at most $1/\log^4 n$-fraction of the windows are blocked in total.
    Recall that the windows blocked by the real algorithm are a subset, the remaining at least $(1-1/\log^4 n)$-fraction of the windows are guaranteed to not be blocked by $j'\neq j$, regardless of the randomness used by hash function $j$.

    For $j$ fixed, the algorithm records the first $(j, v)$ such that $\checkfn(j, v)$ returns YES.
    We will prove 
    \begin{itemize}
        \item for all $v$ such that $v\neq h^{(j)}(y)$ for any $\epsilon/256$-heavy hitter $y$, $\checkfn(j, v)$ will return YES with probability $O(K^{-1/3})$;
        \item for every $\epsilon$-heavy hitter $x^*$, $\checkfn(j, v)$ will be initiated for $v=h^{(j)}(x^*)$ and will return YES with probability $\Omega(\epsilon)$.
    \end{itemize} 

    For the first bullet point, consider the last round of $\checkfn$, i.e., $l=\maxl$.
    Consider a node $I$ in level $l$ of the execution forest, the length of $I$ is $2^l=C\sqrt{\epsilon n}/8$.
    Fix $v$ such that $v\neq h^{(j)}(y)$ for any $\epsilon/256$-heavy hitter $y$, we now upper bound the probability that $\checkfn(j, v)$ passes the last round on $I$.
    First observe that any element in $U_{I,j,v}$ that is not an $\epsilon/256$-heavy hitter has overall frequency that is less than $2^{l-1}=C\sqrt{\epsilon n}/16$.
    Hence, by Lemma~\ref{lem:KSHH_no}, this requires at least two elements in $U_{I,j,v}$ with frequency at least $2^l\cdot K^{-1/8}=C\sqrt{\epsilon n}K^{-1/8}/8$.
    But there can be at most $O(K^{1/4}/\epsilon)$ elements with this frequency, even in the whole stream.
    The probability there exists two of them with hash value $v$ is at most $O\left(\left(\frac{K^{1/4}}{\epsilon K}\right)^2\right)\le O\left(K^{-4/3}\right)$.
    In particular, $\checkfn(j, v)$ passes the last round on $I$ with probability at most $K^{-4/3}$.
    There are $O(1)$ nodes in the last level, taking a union bound over $v$ and the nodes in this level, we conclude that the algorithm records a wrong hash value for $j$ with probability at most $O(K^{-1/3})$.
    In the remainder of the proof, we will show that we will sample and start $\checkfn(j, v)$ from a ``good'' window with constant probability, and a $\checkfn(j, v)$ that is started from a ``good'' window will return YES with constant probability.

    Consider a window $i$, suppose we start $\checkfn$ in parallel at the end of this window.
    Let $I_0,I_1,\ldots,I_{\maxl}$ be the nodes in the execution forest that this check will run $K$-SHH on, where $I_l$ is in level $l$ (in particular, $I_0$ consists of window $i+1$).
    We say window $i$ is ``good'' if
    \begin{enumerate}
        \item $I_0,\ldots,I_{\maxl}$ exist, i.e., $i$ is not close to the end of the stream,
        \item for each $l=0,\ldots,\maxl$, $I_l$ has at most $2^l\log^{-1} n$ blocked windows, and
        \item for each $l=0,\ldots,\maxl$, the second moment of $I_l$ is at most $2^lW\log^3 n$, i.e., $\sum_x f_{x, I_l}^2\leq 2^lW\log^3 n$.
    \end{enumerate}
    Note that the definition of a good window depends on the randomness used when checking other hash functions $j'\neq j$ (for the second condition), but not the randomness used by $j$.

     For each good window, $v=h^{(j)}(x^*)$, and $l=0,\ldots,\maxl$, we have
    \[
        \mathbb{E}\left[\sum_{x\in U_{I_l,j,v},x\neq x^*} f_{x,I_l}^2\right] \leq \Pr[x\in U_{I_l,j,v}] \cdot  2^lW\log^3 n \le (2^l q) \cdot 1/K \cdot 2^lW\log^3 n = 4^l\log^3 n/K.
    \] Therefore, by Markov's inequality and taking the union bound over all $l$, we have that with probability at least $1-o(1)$, for every  $l=0,\ldots,\maxl$, 
    \[ \sum_{x\in U_{I_l,j,v},x\neq x^*} f_{x,I_l}^2 \leq 4^l\log^5 n/K\, .\]

    Assuming the bound holds, if we start $\checkfn(j, v)$ for $v=h^{(j)}(x^*)$ from a good window, then for each round $l$, $K$-SHH returns YES with probability $1-1/\poly\, n$ by Lemma~\ref{lem:Yes_case}.
    Taking a union bound over all levels, $\checkfn$ returns YES except with $o(1)$ probability.
    It remains to show that we will start such a $\checkfn$ with constant probability.
    
    We show that at least a constant fraction of the windows are good.
    \begin{claim}\label{clm:good_window}
         With probability $1-1/\poly\, n$, at least $1/2-o(1)$ fraction of the windows are good.
    \end{claim}
    To see this, we analyze the fraction of windows that is \emph{not} good due to the violation of each of the conditions above.
    \begin{enumerate}
        \item Note that $2^{\maxl}=C\sqrt{\epsilon n}/8=n/8W$.
        As long as $i+4\cdot 2^{\maxl}\leq n/W$, i.e., $i\leq n/2W$, $I_0,\ldots,I_{\maxl}$ exist.
        At most half of the windows violate this condition.
        \item By applying \cref{lem:block_window_concentrate} and taking the sum over all hash functions $j'\ne j$, we observe that at most $O\left((n/W)\cdot \log^{-3} n\right)$ windows are blocked with probability $1-1/\poly\, n$. For a fixed $l$, at most $\frac{O((n/W)\cdot \log^{-3} n)}{2^l\log^{-1}n}=O\left((n/2^lW)\log^{-2} n\right)$ nodes in level $l$ can violate the condition.
        For each such node, there are exactly $2^l$ windows $i$ that could lead to this node in level $l$.
        At most $O((n/W)\log^{-2} n)$ windows violate this condition for each $l$, thus, a total of $O((n/W)\log^{-1} n)$ windows violate this condition.
        This is $o(1)$-fraction.
        \item Similarly, since the second moment of the stream is $O(n)$, fix $l$, at most $\frac{O(n)}{2^lW\log^3 n}=O((n/2^lW)\log^{-3} n)$ nodes in level $l$ can violate this condition.
        At most $O((n/W)\log^{-3} n)$ windows violate this condition for each $l$.
        Thus, a total of $o(1)$-fraction windows violate this condition.
    \end{enumerate}
    Given the claim, we now prove that with $\Omega(1)$ probability, the (actual) algorithm will start $\checkfn(j,v)$ for the correct $v$ from some ``good'' window.
    Concretely, we prove the following claim.
    \begin{claim}\label{clm:good_start}
        With $\epsilon/25$ probability, there exists a \emph{good} window $i$ such that
        \begin{itemize}
            \item $x^*$ is in $S_i^{(j)}$, and it is the only element in $S_i^{(j)}$ that appears in window $i$;
            \item at the end of window $i$, at most $100$ $\checkfn(j,\cdot)$ are still running.
        \end{itemize}
    \end{claim}
    It is easy to verify that when the event in the claim happens, we start $\checkfn(j, v)$ for the correct $v$ from that window, and hence, the lemma holds.
    
    To prove the claim, we need to prove that this happens in a good window and when there are at most $100$ other $\checkfn$ running.
    We apply Lemma~\ref{lem:bound_checkfn_time}, which asserts that at most $1/4+o(1)$ of the windows in expectation can have more than $100$ $\checkfn(j,\cdot)$ running.
    Combined with Claim~\ref{clm:good_window}, in expectation, at least $1/4-o(1)$-fraction of the windows are simultaneously good and have at most $100$ $\checkfn$ running.
    Like in \cref{sec:SHH}, we will prove the following claim by induction, which implies the bound we aim for.
    
     Let $G_i$ be the random variable indicating the number of windows that are simultaneously good and have at most $100$ $\checkfn$ running \emph{from window $i+1$ to the end}.
     We have $\E[G_0]\geq (1/4-o(1))n/W$.
     Let $R_i$ be the randomness used by the algorithm in window $i$, i.e., the random sets $S^{(j)}_1,\ldots,S^{(j)}_i$, and all the random bits used by $K$-SHH in window $i$.
   
    \begin{claim}
     Fix all randomness used by hash functions $j'\neq j$.
     Let $E$ be the (bad) event that we \emph{don't} have a window $i$ that is simultaneously good and have at most $100$ $\checkfn$ running such that $x^*$ is in $S_{i}^{(j)}$, and it is the only element in $S_{i}^{(j)}$ that appears in the window.
     Then we have
     \[
        \Pr[E\mid h^{(j)}, R_{\leq i}] \leq \gamma\left(\frac{q}{2}\cdot \E[G_i\mid h^{(j)}, R_{\leq i}]\right),
     \]
     where $\gamma(y):=\exp(-y/3-5y^2/9)$.
    \end{claim}

    The claim above is the exact same claim that we proved in \cref{sec:SHH}. And as we observed earlier, the function $\gamma$ on the RHS is concave.

   By setting $i=0$.
    We have that 
    \[
        \Pr[E\mid h^{(j)}]\leq \gamma\left(\frac{q}{2}\cdot \E[G_0\mid h^{(j)}]\right).
    \]
    Taking expectation over $h^{(j)}$ and using the concavity and monotonicity of $\gamma$, we have
    \begin{align*}
        \Pr[E]&\leq \E_{h^{(j)}}\left[\gamma\left(\frac{q}{2}\cdot \E[G_0\mid h^{(j)}]\right)\right] \\
        &\leq \gamma\left(\frac{q}{2}\cdot \E[G_0]\right) \\
        &\leq \gamma\left(\frac{q}{2}\cdot (1/4-o(1))n/W\right) \\
        &\leq \gamma\left(\epsilon (1/8-o(1))\right) \\
        &\leq \exp(-\epsilon/25-\epsilon^2/64) \\
        &\leq 1-\epsilon/25\, .
    \end{align*}
    That is, the bad event $E$ happens with probability at most $1-\epsilon/25$, we must start $\checkfn$ on the correct hash value $v=h^{(j)}(x^*)$ with probability at least $\epsilon/25$.
    This implies Claim~\ref{clm:good_start}.

We finally analyze the space usage and correctness of \cref{alg:ell2_eps}.

\begin{theorem}\label{thm:eps_correctness_space}
    The space usage of Algorithm~\ref{alg:ell2_eps} is $O(\log n/\epsilon)$ bits.
    It finds all the $\epsilon$-heavy hitters with probability $1-\left(2^{-\Omega(\log n/(\log\log n-\log \epsilon))}\right)/\epsilon$\footnote{This probability is at least a constant for $\epsilon=\Omega\left(\frac{1}{2^{\sqrt{\log n}}}\right)$.} and it outputs an element that is not an $\epsilon/256$-heavy hitter with probability at most $1/U^{1-o(1)}$.
\end{theorem}
\begin{proof}
   It is straightforward to verify that the space usage is $O(\log n/\epsilon)$ bits: at most $O(J)=O(\log n/\epsilon\log K)$ $\checkfn$ and $K$-SHH may be running in parallel, each uses $O(\log K)$ bits regularly; at most $O(1/\epsilon)$ instances of $K$-SHH may use $O(\log n)$ bits at the same time.

    Now we lower bound the probability that it finds the correct set of $\epsilon$-heavy hitters. Consider an $\epsilon$-heavy hitter $x^*$.
    Let $X_j$ indicate if we find $v_j=h^{(j)}(x^*)$.
    Lemma~\ref{lem:prob_hash_value} implies that $\Exp[X_j\mid X_1,\ldots,X_{j-1}]\geq \epsilon/25$.
    The expectation of $\sum X_j$ is at least $\epsilon J/25=\frac{128}{25}\log U/\log K$.
    Thus, by \cref{thm:azuma-hoeffing}, $\Pr[\sum_j X_j\geq 4\log U/\log K]\geq 1-2^{-\Omega(\log U/\log K)}=1-2^{-\Omega(\log n/(\log\log n-\log \epsilon))}$.
    When this happens, $x^*$ satisfies the last line of Algorithm~\ref{alg:ell2}. Taking the union bound over at most $1/\epsilon$ many $\epsilon$-heavy hitters, we get the desired bound.

    On the other hand, consider any $x$ that is not an $\epsilon/256$-heavy hitter that matches $4\log U/\log K$ hash values recorded. For each $v_j$ it matches, either $h^{(j)}(x)=h^{(j)}(x^*)$ for some $\epsilon/256$-heavy hitter $x^*$  or $v_j\neq h^{(j)}(x^*)$ for any $\epsilon/256$-heavy hitter $x^*$. 
    The former happens with probability $O((\epsilon K)^{-1})$, and in the latter case, it is the hash value found by $j$ with probability $O(K^{-1/3})$ by Lemma~\ref{lem:prob_hash_value}.
    By union bounding over the choice of $4\log U/\log K$ indices $j$, the probability that this happens is at most
    \[
        \binom{J}{4\log U/\log K}\cdot O(K^{-1/3})^{4\log U/\log K}\leq O(K^{-1/3})^{4\log U/\log K}\leq 1/U^{2-o(1)}.
    \]
    Thus, the probability that there exists an element that is not an $\epsilon$-heavy hitter that also satisfies the last line of the algorithm is at most $1/U^{1-o(1)}$.
    This proves the theorem.
\end{proof}

\subsection{Proof of \cref{thm:main result 2}}
\cref{alg:ell2_eps} assumed that both the length of the stream and the second moment are known in advance. It also assumed that the heavy hitters are evenly distributed in the stream. We first show how to remove the assumption that the heavy hitters are evenly distributed in the stream. We claim the following regarding the distribution of heavy hitters in partially random order streams.
\begin{claim}\label{claim:heavy-hitter-dist}
For each $\epsilon$-heavy hitter, every $L$ consecutive windows have at least $2L/3$ occurrences of this heavy hitter with probability $1-2\exp(-L/18)$.
\end{claim}
Before we prove the claim, let us see how we can modify our algorithm based on this claim. 
We first modify \cref{Algorithm 1} so that it searches for an element that occurs in many windows instead of all windows.
For $L\ge K^{3/16}$, we modify \cref{alg1_check_line-14} in the K-SHH algorithm (\cref{Algorithm 1}) to look for an occurrence of the sampled element $y$ in every $90\log n$ windows, rather than every window, and do this check in the next $90\log n L/K^{1/8}$ windows. Note that this increases the total number of windows using $\log n$ bits in \cref{lem_KSHH_expectation,lem:low_freq_conc,lem:high_freq_conc} only by an $O(\log n)$ factor, which does not affect how these lemmas are applied in \cref{sec:SHH,sec:epsilon-heavy}, since $K$ is a large polynomial in $\log n$. It is also easy to see that we get the same success probability in the YES case (\cref{lem:Yes_case}) because of \cref{claim:heavy-hitter-dist} (and the fact that $L\gg\log n$). There is no change in the analysis of the NO case (\cref{lem:KSHH_no}). 

For $L<K^{3/16}$, we use a similar algorithm as before, where we hash every element in $V$ using a uniform $O(\log K)$ bit hash function and return YES if there are at least $2L/3$ elements in the substream restricted to $V$ and all their hash values are equal, else return NO. We will now argue that \cref{lem_KSHH_expectation,lem:low_freq_conc,lem:high_freq_conc,lem:KSHH_no} hold in this case, and \cref{lem:Yes_case} holds \emph{after changing the premise} to an $x^*$ occurring in $2L/3$ of the windows. Observer that \cref{lem_KSHH_expectation,lem:low_freq_conc,lem:high_freq_conc} hold trivially because the algorithm never uses more than $O(\log K)$ bits.
In the YES case corresponding to \cref{lem:Yes_case}, the condition that $\sum_{y\ne x^*\in V} f_y^2 \le L^2 \log^5 n/K$ implies that there is no element other than $x^*$ in the substream restricted to $V$. 
Therefore, the algorithm outputs YES in this case. On the other hand, consider the NO case. If the substream restricted to $V$ has fewer than $2L/3$ elements, then the algorithm always outputs NO. If the substream has at least $2L/3$ elements, then the conditions in the NO case imply that there are at least $K^{1/8}/6$ distinct elements in the substream. The probability that they all have the same hash value is at most $1/\poly\, n$.

To accommodate for the new premise of \cref{lem:Yes_case}, we will need to update the proof of \cref{lem:prob_hash_value,lem:prob_hash_value_eps}, which are the only places we used \cref{lem:Yes_case}.
We change the definition of a ``good'' window to further require for each $l\geq 10$, the node $I_l$ has at least $2^{l+1}/3$ occurrences of the heavy hitter $x^*$.
By \cref{claim:heavy-hitter-dist}, all but an $2\exp(-50)$-fraction of the windows satisfy this extra condition.
Next, we modify the algorithm so that it starts from iteration $l=10$.
We also increase the number of $\checkfn$ allowed at the same time by a factor of $2^{10}$.
The proofs still go through with a slightly smaller constant fraction of good windows.
The correctness of \cref{alg:ell2,alg:ell2_eps} still hold and the space usage is also unaffected.
We now prove \cref{claim:heavy-hitter-dist}.
\begin{proof}[Proof of \cref{claim:heavy-hitter-dist}]
Fix an $\epsilon$-heavy hitter $x^*$. By definition, the frequency of $x^*$ is at least $C\sqrt{\epsilon n}$ and the positions of its occurrences are marginally uniform in the stream. Equivalently, we can view the distributions of its occurrences in the stream as sampling at least $k=n/W$ balls from $n$ bins without replacement, where $n$ is the length of the stream and $W$ is the length of the window we chose in \cref{alg:ell2_eps}. Let $L\leq n/W$ and $j\in [0,n/W-L]$. Let $X_{L,j}$ denote the number of occurrences of $x^*$ in windows ranging from $j+1$ to $j+L$. Applying \cref{thm:sampling without replacement} with $\ell = LW$ and $i = jW$, we get that
\[
\Pr[X_{L,j} < 2L/3] \le 2\exp(-\ell k/18n)\leq 2\exp(-L/18) \, .
\]
This proves the claim.
\end{proof}

We now show how to remove the assumption about the length of the stream. It is not hard to show that if we replace $n$ with $\tilde{n} \in [0.9 n, 1.1 n]$ and if \cref{alg:ell2_eps} executes only on a $0.9$ fraction of the stream, the analysis still works. So, we guess the value of $n$ in increasing powers of $(1.1)$, i.e., $(1.1),(1.1)^2,\dots$, and ensure that whenever we guess the correct value of $n$, at least a $0.9$ fraction of the stream is still left. At each time step, the algorithm maintains $25$ possible guesses for $n$ in increasing order $(1.1)^i,\dots,(1.1)^{i+24}$ and maintains a counter $\mathsf{len}$ to measure the current length of the stream. At the beginning of the stream, $i$ is set to be $0$, and we start executing \cref{alg:ell2_eps} in parallel for each of these guesses. Whenever $\mathsf{len}$ reaches $(1.1)^i +1$, we discard the first guess, along with its corresponding execution of \cref{alg:ell2_eps}, and we add a new guess $(1.1)^{i+25}$ and begin executing an instance of \cref{alg:ell2_eps} for this guess, starting from this point in the stream. At the end of the stream we choose the guess $(1.1)^i$ such that $(1.1)^{i-1}< n\le (1.1)^i$ and output whatever the \cref{alg:ell2_eps} corresponding to this guess outputs and discard the rest. Note that the length $\tilde{n}$ of the stream guessed for this instance of \cref{alg:ell2_eps} lies in the range $ [0.9 n, 1.1 n]$, where $n$ is the true length of the stream. In addition, at the point where the guess $(1.1)^i$ is added, at most $(1.1)^{i-25}$ elements in the stream have been seen. Hence, for the correct guess, \cref{alg:ell2_eps} would have executed on at least a $1-(1.1)^{-25} > 0.9$ fraction of the remaining portion of the stream, and hence, we would get the same guarantees as in \cref{thm:eps_correctness_space}. Finally, we note that this meta algorithm has the same space complexity as \cref{alg:ell2_eps}.


\ifnum\doubleblind=0
\section*{Acknowledgment}
We thank Uma Girish for her contributions in the early stages of this project.
\textsc{S.V.} was supported in part by NSF award CCF 2348475. Part of the work was conducted when \textsc{S.V.} was visiting the Simons Institute for the Theory of Computing as a research fellow in the Sublinear Algorithms program. Part of the work was conducted when \textsc{S.V.} was a graduate student at Harvard University, and supported in part by a Google Ph.D. Fellowship, a Simons Investigator Award to Madhu Sudan, and NSF Award CCF 2152413. Part of the work was conducted when \textsc{S.V.} was visiting Princeton Univeristy as an exchange student under the IvyPlus Exchange Scholar program.
Huacheng Yu is supported in part by NSF CAREER award CCF-2339942.
\fi

\printbibliography
\end{document}